\documentclass[runningheads]{llncs}
\usepackage[T1]{fontenc}
\usepackage{graphicx}
\usepackage{amsmath,amssymb,amsfonts}
\usepackage{mathtools}
\usepackage{thmtools,thm-restate}
\usepackage{listings}
\usepackage{pgfplots}
\usepackage{subfig}
\usepgfplotslibrary{fillbetween}
\usepgfplotslibrary{units} 
\usepackage{tikz}
\usetikzlibrary{spy}
\usetikzlibrary{plotmarks}
\usetikzlibrary{shapes.geometric, arrows, shapes, backgrounds}

\usepackage{xcolor}
\usepackage{optidef}
\usepackage[linesnumbered,ruled, lined]{algorithm2e}
 \usepackage{lipsum}

\tikzstyle{startstop} = [rectangle, rounded corners, minimum width=3cm, minimum height=1cm,text centered, draw=black, fill=red!30]
\tikzstyle{io} = [trapezium, trapezium left angle=70, trapezium right angle=110, minimum width=3cm, minimum height=1cm, text centered, draw=black, fill=blue!30]
\tikzstyle{process} = [rectangle, minimum width=3cm, minimum height=1cm, text centered, draw=black, fill=orange!30]
\tikzstyle{decision} = [diamond, minimum width=3cm, minimum height=1cm, text centered, draw=black, fill=green!30]
\tikzstyle{arrow} = [thick,->,>=stealth]

\newcommand{\tikzDefaultsCom}[0]{
  \pgfplotsset{
   every axis/.append style={line width = 0.5pt},
    title style={align=center},
    xlabel={$v/C_{max}$},
    ylabel={Percentage of Cracked Passwords},
    log basis x={10},
    grid=major,
    cycle list = {{black}, {red}, {blue, densely dashed}, {yellow, densely dashdotted}, {blue, densely dashdotdotted}, {blue!50!black, loosely dotted}},
    legend style = {font=\tiny, at={(.01,.99)}, anchor=north west},
    legend entries = {deterministic\\ Time-Even Breakpoints\\ Cost-Even Breakpoints\\ $m=7$\\ $m=99$\\ improvement:{\color{black} black}-{\color{blue}blue}\\}
    }
}

\newcommand{\tikzDefaultsCost}[0]{
  \pgfplotsset{
   every axis/.append style={line width = 0.5pt},
    title style={align=center},
    xlabel={$v/C_{max}$},
    ylabel={Percentage of Cracked Passwords},
    log basis x={10},
    grid=major,
    cycle list = {{black}, {pink, densely dotted}, {red, densely dashed}, {yellow, densely dashdotted}, {blue, densely dashdotdotted}, {blue!50!black, loosely dotted}},
    legend style = {font=\tiny, at={(.01,.99)}, anchor=north west},
    legend entries = {deterministic\\ $m =2$\\ $m=3$\\ $m=7$\\ $m=99$\\ improvement:{\color{black} black}-{\color{blue}blue}\\}
    }
}

\newcommand{\tikzDefaultsTime}[0]{
  \pgfplotsset{
    every axis/.append style={line width = 0.1pt},
    title style={align=center},
    xlabel={$v/C_{max}$},
    ylabel={Percentage of Cracked Passwords},
    log basis x={10},
    grid=major,
    cycle list = {{black}, {red}, {red, mark = x}, {yellow, densely dashdotted}, {blue, densely dashdotdotted}, {blue!50!black, loosely dotted}},
    legend style = {font=\tiny, at={(.01,.99)}, anchor=north west},
    legend entries = {Deterministic, $m =3$, Fails Optimality Test, time-even $m=3$, time-even $m=7$, time-even $m=99$, improvement:{\color{black} black}-{\color{blue}blue}}
    }
}

\newcommand{\tikzDefaultsBite}[0]{
  \pgfplotsset{
    every axis/.append style={line width = 0.5pt},
    title style={align=center},
    xlabel={$v/C_{max}$},
    ylabel={Percentage of Cracked Passwords},
    log basis x={10},
    grid=major,
    cycle list = {{black}, {pink, densely dotted}, {red, densely dashed}, {yellow, densely dashdotted}, {blue, densely dashdotdotted}, {blue!50!black, loosely dotted}},
    legend style = {font=\tiny, at={(.01,.99)}, anchor=north west},
    legend entries = {deterministic\\  $m = 3$, uniform\\ $m = 3$, biteopt\\ $m = 7$, uniform\\ $m = 7$, biteopt\\}
    }
}

\DeclareMathOperator*{\argmax}{arg\,max}

\spnewtheorem{observation}{Observation}{\bfseries}{\itshape}

\newcommand{\defeq}{\vcentcolon=}
\newcommand{\MHF}{\mathsf{MHF}}
\newcommand{\peak}{\mathsf{peak}(\vec{q})}
\newcommand{\ma}[2]{\Delta\left(#1,#2\right)}
\newcommand{\op}[1]{\oplus_{i=1}^{#1} \varpi_i(1, \tau_i)}
\newcommand{\vp}[1]{\varpi_{#1}(1, \tau_{#1})}
\newcommand{\margin}[2]{\Delta\left(#1,#2\right)}
\newcommand{\mar}[2]{\Delta\left(#1\,\middle|\,#2\right)}
\newcommand{\con}[2]{\Delta^{\circ}\left(#1\,\middle|\,#2\right)}
\newcommand{\ins}[2]{\Delta^+\left(#1\,\middle|\,#2\right)}
\newcommand{\bundle}[0]{\varpi_i(j_1,j_2)}
\newcommand{\EXC}[1]{\mathsf{ExtendbyConcat}(v,\vec{q},#1)}

\newcommand{\EXI}[1]{\mathsf{ExtendbyInsert}(v, \vec{q}, #1)}
\newcommand{\EX}[0]{\mathsf{Extend}(v,\vec{q})}
\newcommand{\OPTDIS}[0]{\mathsf{FindOptDis}()}
\newcommand{\OT}[1]{\mathsf{OptimalityTest}(v,\vec{q},#1)}
\newcommand{\FG}[1]{\mathsf{FindOptSeq}(v,\vec{q},#1)}
\newcommand{\before}[2]{\mathsf{Before}\left(#1\,\middle|\,#2\right)}
\newcommand{\len}[1]{\mathsf{Len}(#1)}
\newcommand{\lenz}[0]{\mathsf{Len}}
\newcommand{\lenmax}[0]{\mathsf{Len}_{max}}
\newcommand{\unchecked}[1]{\mathsf{unchecked}(#1)}
\begin{document}
\title{Cost-Asymmetric Memory Hard \\Password Hashing}
\author{Wenjie Bai\inst{1} \and
Jeremiah Blocki\inst{1} \and
Mohammad Hassan Ameri\inst{1}}
\authorrunning{W. Bai et al.}
\institute{Purdue University, West Lafayette IN, 47907, USA\\
\email{\{bai104, jblocki, mameriek\}@purdue.edu}}
% \email{\{bai104,jeremiahb, \}@uni-heidelberg.de}
%
\maketitle              % typeset the header of the contribution
\begin{abstract}
In the past decade billions of user passwords have been exposed to the dangerous threat of offline password cracking attacks. An offline attacker who has stolen the cryptographic hash of a user's password can check as many password guesses as s/he likes limited only by the resources that s/he is willing to invest to crack the password. Pepper and key-stretching are two techniques that have been proposed to deter an offline attacker by increasing guessing costs. Pepper ensures that the cost of rejecting an incorrect password guess is higher than the (expected) cost of verifying a correct password guess. This is useful because most of the offline attacker's guesses will be incorrect. Unfortunately, as we observe the traditional peppering defense seems to be incompatible with modern memory hard key-stretching algorithms such as Argon2 or Scrypt. We introduce an alternative to pepper which we call Cost-Asymmetric Memory Hard Password Authentication which benefits from the same cost-asymmetry as the classical peppering defense i.e., the cost of rejecting an incorrect password guess is larger than the expected cost to authenticate a correct password guess. When configured properly we prove that our mechanism can only reduce the percentage of user passwords that are cracked by a rational offline attacker whose goal is to maximize (expected) profit i.e., the total value of cracked passwords minus the total guessing costs. We evaluate the effectiveness of our mechanism on empirical password datasets against a rational offline attacker. Our empirical analysis shows that our mechanism can reduce significantly the percentage of user passwords that are cracked by a rational attacker by up to $10\%$.

\end{abstract}

\section{Introduction}
In the past decade data-breaches have exposed billions of user passwords to the dangerous threat of offline password cracking. An offline attacker has stolen the cryptographic hash $h_u=H(pw_u, salt_u)$ of a target user (u) and can validate as many password guesses as s/he likes without getting locked out i.e., given $h_u$ and $salt_u$\footnote{The salt value protects against pre-computation attacks such as rainbow tables and ensures that the attacker must crack each individual password separately. For example, even if Alice and Bob select the same password $pw_A=pw_B$ their password hashes will almost certainly be different i.e., $h_A = H(pw_A, salt_A) \neq H(pw_B, salt_B) = h_B$ due to the different choice of values and collision resistance of the cryptographic hash function $H$.} the attacker can check if $pw_u=pw'$ by computing $h'=H(pw', salt_u)$ and comparing the hash value with $h_u$.  Despite all of the security problems text passwords remain entrenched as the dominant form of authentication online and are unlikely to be replaced in the near future~\cite{SP:BHVS12}. Thus, it is imperitive to develop tools to deter offline attackers. 

An offline attacker is limited only by the resources s/he is willing to invest in cracking the password and a rational attacker will fix a guessing budget to optimally balance guessing costs with the expected value of the cracked passwords. Key-Stretching functions intentionally increase the cost of the hash function $H$ to ensure that an offline attack is as expensive as possible. Hash iteration is a simple technique to increase guessing costs i.e., instead of storing $(u,salt_u, h_u=H(pw_u,salt_u))$ the authentication server would store $(u,salt_u, h_u=H^t(pw_u,salt_u))$ where $H^{i+1}(x)\defeq H(H^{i}(x))$ and $H^{1}(x)\defeq H(x)$. Hash iteration is the traditional key-stretching method which is used by password hashing algorithms such as PBKDF2~\cite{kaliski2000pkcs} and BCRYPT~\cite{provos1999bcrypt}. Intuitively, the cost of evaluating a function like PBKDF2 or BCRYPT scales linearly with the hash-iteration parameter $t$ which, in turn, is directly correlated with authentication delay. Cryptocurrencies have hastened the development of Application Specific Integrated Circuits (ASICs) to rapidly evaluate cryptographic hash functions such as SHA2 and SHA3 since mining often involves repeated evaluation of a hash function $H(\cdot)$. In theory an offline attacker could use ASICs to substantially reduce the cost of checking password guesses. In fact, Blocki et al.~\cite{SP:BloHarZho18} argued that functions like BCRYPT or PBKDF2 cannot provide adequate protection against an offline attacker without introducing an unacceptable authentication delay e.g., $2$ minutes. 

Memory-Hard Functions (MHFs)~\cite{Per09} have been introduced to address the short-comings of hash-iteration based key-stretching algorithms like BCRYPT and PBKDF2.  Candidate MHFs include SCRYPT~\cite{Per09}, Argon2 (which was declared as the winner of Password Hashing Competition \cite{PHC} in 2015) and DRSample~\cite{CCS:AlwBloHar17}. Intuitively, a password hash function is memory hard if any algorithm evaluating this function must lock up large quantities of memory for the duration of computation. One advantage of this approach is that RAM is an expensive resource even on an ASIC leading to egalitarian costs i.e., the attacker cannot substantially reduce the cost of evaluating the hash function using customized hardware. The second advantage is that the Area-Time cost associated with a memory hard function can scale quadratically in the running time parameter $t$. Intuitively, the honest party can evaluate the hash function $\MHF(\cdot ; t)$ in time $t$, while any attacker evaluating the function must lock up $t$ blocks of memory for $t$ steps i.e., the Area-Time cost is $t^2$. The running time parameter $t$ is constrained by user patience as we wish to avoid introducing an unacceptably long delay while the honest authentication server evaluates the password hash function during user authentication. Thus, quadratic cost scaling is desireable as it allows an authentication server to increase password guessing costs rapidly without necessarily introducing an unacceptable authentication delay.

Peppering~\cite{manber1996simple} is an alternative defense against an offline password attacker. Intuitively, the idea is for a server to store $(u, salt_u, h_u= H(pw_u, salt_u, x_u))$. Unlike the random salt value $salt_u$, the random pepper value $x_u \in [1, x_{max}]$ is {\em not} stored on the authentication server. Thus, to verify a password guess $pw'$ the authentication server must compute $h_1 = H(pw', salt_u, 1), \ldots h_{x_{max}} = H(pw', salt_u, x_{max})$. If $pw'=pw_u$ then we will have $h_{x_u}= h_u$ and authentication will succeed. On the other hand, if $pw' \neq pw_u$ then we will have $h_i \neq h_u$ for all $i \leq x_{max}$ and authentication will fail. In the first case (correct login) the authentication server will not need to compute $h_i = H(pw', salt_u, i)$ for any $i > x_{u}$, while in the second case (incorrect guess) the authentication server will need to evaluate $h_i$ for every $i \leq x_{max}$. Thus, the expected cost to verify a correct password guess is lower than the cost of rejecting an incorrect password guess. This can be a desirable property as a password attacker will spend most of his time eliminating incorrect password guesses, while most of the login attempts sent to the authentication server will be correct.

A natural question is whether or not we can combine peppering with Memory Hard Functions to obtain both benefits: quadratic cost scaling and cost-asymmetry.
\begin{question}
    Can we design a password authentication mechanism that incorporates \textbf{cost-asymmetry} into ASIC resistant Memory Hard Functions while having the benefits of \textbf{fully quadratic cost scaling} under the the constraints of authentication delay and expected workload?
\end{question}

\noindent{\bf Naive Approach:} At first glance it seems trivial to integrate pepper with a memory hard function $\MHF(\cdot)$ e.g., when a new user $u$ registers with password $pw_u$ we can simply pick our random pepper $x_u\in [1,x_{max}]$, salt $salt_u$, compute $h_u=\MHF(pw_u,salt_u,x_u;t)$ and store the tuple $(u,salt_u,h_u)$.
 Unfortunately, the solution above is overly simplistic.  How should this parameter be set? We first observe that the authenication delay for our above solution can be as large as $t \cdot x_{max}$ since we may need to compute $\MHF(pw,salt_u, x; t)$ for every value of $x \in [1,x_{max}]$ and this computation must be carried out sequentially to reap the cost-asymmetry benefits of pepper.  Similarly, the Area-Time cost for the attacker to evaluate $\MHF(pw,salt_u, x; t)$ for every value of $x \in [1,x_{max}]$ would scale with $t^2 \cdot x_{max}$. This may seem reasonable at first glance, but what if the authentication server had not used pepper and instead stored $h_u= \MHF(pw_u,salt_u; t \cdot x_{max})$ using the running time parameter $t' = t \cdot x_{max}$? In this case the authentication delay is identical, but the attacker's Area-Time cost would be $t'^2 = t^2 \cdot x_{max}^2$ --- an increase of $x_{max}$ in comparison to the naive solution. Thus, the naive approach to integrate pepper and memory hard functions loses much of the benefit of quadratic scaling. 

\noindent {\bf Halting Puzzles:}  Boyen ~\cite{USENIX:Boyen07} introduced the notion of a halting puzzle where the ``pepper'' value is replaced with a random running time parameter. In particular, when a new user $u$ registers with a password  $pw_u$  we can pick our random running time parameter $t_u \in [1, t_{max}]$ along with $salt_u$ and store $(u,salt_u, h_u)$ where  $h_u=\MHF(pw_u,salt_u; t_u)$. Given a password guess $pw'$ the authentication server will locate $salt_u, h_u$ and accept if and only if $h_u = \MHF(pw', salt_u; t)$ for some $t \in [1, t_{max}]$. All memory hard functions $\MHF(w; t)$ we are aware of generate a stream of data-labels $L_1,\ldots, L_t$ where  $L_i = \MHF(w; i)$ and $L_{i+1}$ can be computed quickly once the prior labels $L_1,\ldots, L_i$ are all stored in memory. Thus, whenever the user attempts to login with password $pw_u'$ the honest server can simply start computing $\MHF(pw_u', salt_u; t_{max})$ to generate a stream of labels $L_1',L_2',\ldots$ and immediately accept if we find some label $i \leq t$ which matches the password hash i.e., $L_i = h_u$. Observe that whenever the user enters the correct password $pw_u'=pw_u$ the honest authentication server will be able to halt early after just $t_u\leq t_{max}$ iterations. By constrast, the only way to definitely reject an incorrect password $pw_u'$ is to finish computing $\MHF(pw_u', salt_u; i)$. The authentication delay is at most $t_{max}$ and it seems like the attacker's area-time cost will scale quadratically i.e., $t_{max}^2$. Thus,  the solution ostensibly seems to benefit from quadratic cost scaling and cost-asymmetry.

However, we observe that an attacker might not choose to compute the entire function $\MHF(pw', salt_u; t)$ for each password guess. For example, suppose that the running time parameter $t_u$ is selected uniformly at random in the range $[1,t_{max}]$, but for each password guess $pw'$ in the attacker's dictionary the attacker only computes $\MHF(pw', salt_u; t_{max}/3)$. The attacker's area-time cost per password guess ($t_{max}^2/9$) would decrease by a factor of $9$, but the attacker's success rate only diminishes by a factor of $1/3$ --- the probability that $t_u \in [1,t_{max}/3]$. Motivated by this observation there are several natural questions to ask. First, can we model how a rational offline attacker would adapt his approach to deal with halting puzzles? Second, if $t_u$ is picked uniformly at random is it possible that the solution could have an adverse impact i.e., could we unintentionally {\em increase} the number of passwords cracked by a rational (profit-maximizing) attacker? Finally, can we find the optimal distribution over $t_u$ which minimizes the success rate of a rational offline attacker subject to constraints on (amortized) server workload and maximum authentication delay.

\subsection{Our contributions}
We introduce Cost-Asymmetric Memory Hard Password Hashing, an extention of Boyen's halting puzzles which can only {\em decrease} the number of passwords cracked by a rational password cracking attacker. Our key modification is to introduce cost-even breakpoints as random running time parameters i.e., we fix $m$ values $t_1 \leq \ldots \leq t_m=t$ such that $t_m^2 = t_i^2 (m/i)$ for all $1 \leq i < m$. Now instead of selecting $x_u$ randomly in the range $[1,t]$ (time-even breakpoints) we pick $x_u \in \{t_1,\ldots, t_m\}$. We can either select $x_u \in \{t_1,\ldots, t_m\}$ uniformly at random or, if desired, we can optimize the distribution in an attempt to minimize the expected number of passwords that the adversary breaks. Then the authentication server computes $h_u=\MHF(pw_u,salt_u; x_u)$ and store the tuple $(u,salt_u,h_u)$ as the record for user $u$.

We adapt the Stackelberg game theoretic framework of Blocki~and Datta\cite{BlockiD16} to model the behavior of a rational pasword cracking attacker when the authentication server uses Cost-Asymmetric Memory Hard Password Hashing. In this model the attacker obtains a reward $v$ for every cracked password and will choose a strategy which maximizes its expected utility --- expected reward minus expected guessing costs. One of the main challenges in our setting is that the attacker's action space is exponential in the size of the support of the password distribution. For each password $pw$ the  attacker can chose to ignore the password, partially check the password or competely check the password. We define efficient algorithms to find a locally optimal strategy for the attacker and identify conditions under which the strategy is also a global optimum (these conditions are satisfied in almost all of our empirical experiments). We can then use black-box optimization to search for a distribution over $x_u$ which minimizes the number of passwords cracked by our utility maximizing attacker. 

When $x_u \in \{t_1,\ldots, t_m\}$ is selected uniformly at random we prove that cost-even breakpoints will only reduce the number of passwords cracked by a rational attacker. By contrast, we provide examples where time-even breakpoints increases the number of passwords that are cracked --- some of these examples are based on empirical password distributions.

We empirically evaluate the effectiveness of our mechanism with 8 large password datasets. Our analysis shows that we can reduce the fraction of cracked passwords by up to 10\% by adopting  cost-asymmetric memory hard password hashing with cost-even breakpoints sampled from uniform distribution. In addition, our analysis demonstrates that the benefit of optimizing the distribution over $x_u$ is marginal. Optimizing the distribution over the breakpoints $t_1,\ldots, t_m$ requires us to accurately estimate many key parameters such as the attacker's value $v$ for cracked passwords and the probability of each password in the user password distribution. If our estimates are inaccurate then we could unintentionally increase the number of cracked passwords. Thus, we recommend instantiating Cost-Asymmetric Memory Hard Password Hashing with the uniform distribution over our cost-even breakpoints $t_1,\ldots, t_m$ as a \emph{prior independent} password authentication mechanism. 

\subsection{Related work}
Trade-off between usability and security lie in the core of mechanism design of password authentication. 
Users tend to pick low-entropy passwords \cite{SP:Bonneau12}, leaving their accounts insecure. Convincing them to select stronger passwords is a difficult task~\cite{campbell2011impact,Komanduri2011,Shay2010,Stanton2005,Inglesant2010,Shay2014}. Password strength meters~\cite{Komanduri2014,Ur2012,NDSS:CarMan14} are commonly embedded in website in the hope that users would select stronger passwords after the strength of their original passwords being displayed. However, it is found that users are often not persuaded by the suggestion of password strength meters \cite{Ur2012,NDSS:CarMan14}. In order to encourge users to pick high-entropy passwords some sites mandate users to follow stringent guidelines when users create their passwords. However, it has been shown that these methods suffer from usability issues \cite{Inglesant2010,NIST2014,Florencio2014lisa,Adams1999}, and in some cases can even lead to users selecting weaker passwords \cite{blockiPasswordComposition,Komanduri2011}. 

Password offline attacks have been a concern since the Unix system was devised~\cite{morris1979password}. Various approaches are developed to expedite the cracking process by the adversary or model password guessability by the hoesty party. Tools like Hashcat~\cite{hashcat} and John the Ripper~\cite{JohnTheRipper} enumerate combinations of tokens as dictionary candidates and are widely used by real-world attackers. Liu et al.~\cite{liu2019reasoning} analyzed these tools using techniques of rule inversion and guess counting to retrive guessing number without explicit enumeration.  Probabilistic models like Probabilistic Context-Free Grammars~\cite{SP:WAMG09,SP:KKMSVB12,NDSS:VerColTho14}, Markov models~\cite{NDSS:CasDurPer12,Castelluccia2013,SP:MYLL14,USENIX:USBCCKKMMS15} have been applied and analyzed in password cracking. Character-level text generation with Long-Short Term Memory (LSTM) recurrent neural networks is fast, lean and accurate in modeling password guessability~\cite{USENIX:MUSKBCC16}. 

Memory-hard functions (MHF) is a key cryptographic primitive. Evaluation of MHF requires large amount of memory in addition to longer computation time, making parallel computation and customized hardware futile to speed up computation process.  Candidate MHFs include SCRYPT~\cite{Per09}, Balloon hashing~\cite{AC:BonCorSch16}, and Argon2~\cite{Argon2} (the winner of the Password Hashing Competition\cite{PHC}). MHFs can be classified into two distinct categories or modes of operation - data-independent MHFs (iMHFs) and data-dependent MHFs(dMHFs) (along with the hybrid idMHF, which runs in both modes). dMHFs like SCRYPT are maximally memory hard~\cite{EC:ACPRT17}, but they have the issue of possible side-channel attacks. iMHFs, on the other hand,  can resist side-channel attakcs but the aAT (amortized Area Time) complexity is at most $ \mathcal{O} (N^2 \log \log N/ \log N)$.  In construction of a iMHF Alwen and Blocki showed that depth robustness, a property associated with a DAG, is both necessary~\cite{C:AlwBlo16} and sufficient~\cite{EC:AlwBloPie17}. Recent work has proposed candidate iMHF constructions that show resistance to currently-known attacks~\cite{blocki2019data}.

\section{Background and Notations}
\paragraph{Password Dataset.}
We use $\mathbb{P}$ to denote the set of all possible passwords,  the corresponding distribution is $\mathcal{P}$. The process of a user $u$ choosing a password for his/her account can be viewed as a random sampling  from the underlying distribution $pw_u\overset{\$}{\leftarrow} \mathcal{P}$. Given a password dataset $D$ of $n_a$ accounts, we can obtain empirical distribution $\mathcal{D}_e$ by approximating $\Pr_{pw_i\sim \mathcal{D}_e}[pw_i] = \frac{f_i}{n_a}$, where $f_i$ is the frequency of $pw_i$ and $n_a$ is the number of accounts present in $D$. Often the empirical distribution can be represented in compact form by grouping passwords with the same frequency into an \emph{equivalence set} i.e., $D_{es}=\{(f_1,s_1),\ldots,(f_i,s_i),\ldots,(f_{n_e},s_{n_e}) \}$, where $s_i$ is the number of passwords which appear with frequency $f_i$ in $D$ and $n_e$ is the total number of equivalence sets. We use $es_i = (f_i,s_i)$ to describe the $i$th equivalence set. In empirical experiments it is often most convenient to work with the compact representation $D_{es}$ of our distribution. We also use $n_p$ to denote the number of distinct passwords in our dataset $D$. Observe that for any dataset we have $n_a \geq n_p \geq n_e$. In fact, we will typically have $n_a\gg n_p \gg n_e$.

\paragraph{Computation Cost of MHF.}

The evaluation of MHF generates a sequence of labels $L_1, L_2,\ldots$, in order to compute $L_i$, some of its predecessor are required. Each label is associated with a running time parameter. We use $T = \{t_1,t_2,\ldots,t_m\}$ to denote the set of possible running time breakpoints and  $q_i = \Pr[t_i]$ to denote the probability that $t_i$ is chosen for user's password $pw_u$. We model the (amortized) Area-Time cost of evaluating $\MHF(\cdot;t)$ as $c_Ht+c_M t^2$, where $c_H$ and $c_M$ are constants. Intuitively, $c_H$ denotes the area of a core implementing the hash function $H$ and $c_M$ represents the area of an individual cell with the capacity to hold one data-label (hash output). Since the memory cost tend to dominate, we ignore the hash cost as simply model the cost as $c_M t^2$. 

% \jnote{This is not the right place to define time-even or cost-even breakpoints! In particular, we are claiming cost-even breakpoints as a new contribution so it should not be introduced in the background section}

% \paragraph{Time-Even/Cost-Even Breakpoints.}
% Define $\beta_i = t_i/t_1, \, \forall i$, halting puzzle approach implicitly selects $\beta_i = i$, which we call \emph{time-even breakpoints}. The naming intuition is that the computation time from $i$th label to $(i+1)$th label is the same. On the other hand, we recommand to use cost-even breakpoints i.e., $\beta_i=\sqrt{i}$ so that the computation cost by advancing to the next label remains the same.

\section{Defender's Model}
In this section, we present the model of the defender. In particular, we show how passwords are stored and verified on the authentication server. 

\paragraph{Account Registration.}
When a user $u$ proposes password $pw_u$ for his/her account at the time of registration, the authentication would randomly generate a salt value $salt_u$, sample a running time parameter $t_u$ from $T$, calculate the memory hard function with parameter $t_u$, i.e., $h_u = \MHF(pw_u,salt_u;t_u)$. In the end, the tuple $(u,salt_u, h_u)$ is stored  in the server as the record of $u$ while $t_u$ is  discarded.
\paragraph{Password Verification.}
Later on, when user $u$ logs in his/her account by submitting $(u,pw_u^{\prime})$ to the server interface, the authentication server would first retrieve record $(u,salt_u, h_u)$, calculate $h_1 = \MHF(pw_u^{\prime},salt_u;t_1)$ and compare $h_1$ with $h_u$. It they are equal, login request is granted. Otherwise, the server would continue to calculate $h_2 = \MHF(pw_u^{\prime},salt_u;t_2)$, compare $h_2$ with $h_u$, so on and so forth. If any of $h_i$ matches $h_u$, then user $u$ successfully logs in his/her account. However, If none of $h_i$ matches $h_u$, the login request is rejected. 

% The formal procedure is presented in Algorithm \ref{alg:verification}.
% \begin{algorithm}[thb]\small
% \caption{Password Verification}
% \label{alg:verification}
% \begin{algorithmic}[1]
% \REQUIRE{$(u,pw_u^{\prime})$}
% \STATE $(u,salt_u, h_h)\leftarrow \mathsf{findRecord}(u)$
% \FOR{$i=1:m$}
% 	\STATE $h_i = \mathsf{KDF}(pw_u^{\prime},salt_u;t_i)$ 
% 	\IF{$h_i == h_u$}
% 	\RETURN Success
% 	\ENDIF
% \ENDFOR
% \RETURN Fail
% \end{algorithmic}
% \end{algorithm}

\paragraph{Workload Constraint}.
The server is  subject to maximum workload constraint which means that  the expected cost cannot exceed server's workload ceiling $C_{max}$. Namely, 
\begin{equation}\small
\sum_{i = 1}^m  q_ic_Mt_i^2\leq C_{max}.
\end{equation}
With probability $q_i$ the server will choose $t_i$, evaluate $\MHF(\cdot;t_i)$  incurs cost $c_Mt_i^2$, thus the expected cost of verifying a password is $\sum_{i = 1}^m  q_ic_Mt_i^2$.

\section{Attacker's model}

In this section,  we first state the assumptions we use in our economic analysis. Then we show how a rational attacker who steals the password hashes from the server would run a dictionary offline attack. Finally, we present the Stackelberg game in modeling the interaction between the defender and the attacker within the framework of ~\cite{BlockiD16}.

\subsection{Assumptions of Economics Analysis}
We made several assumptions about the attacker which facilitates our economic analysis of the attacker's behavior.
\begin{itemize}
\item \emph{rationality}. A rational attacker's objective is to extract as much monetary reward as possible from the cracking process, the only constraint is the number of passwords the attacker is willing to guess.
\item \emph{knowledge}. The attacker knows password distribution $\mathcal{P}$. In practice, password dictionary list can be built using Probabilistic Context-Free Grammars (PCFGs)~\cite{SP:WAMG09,SP:KKMSVB12,NDSS:VerColTho14}, $n$-gram Markov Models~\cite{NDSS:CasDurPer12,Castelluccia2013,SP:MYLL14,USENIX:USBCCKKMMS15} and Neural Networks~\cite{USENIX:MUSKBCC16}.
\item \emph{untarget}. We assume the attacker is untargeted, treats all users without distinction and has universal password value expectation $v$ for all accounts. One can derive a range of estimates for $v$ based on black market studies e.g., Symantec reported that passwords generally sell for \$4---\$30~\cite{passwordBlackMarket} and \cite{stockley_2016} reported that Yahoo! e-mail passwords sell for $\approx \$1$.
\end{itemize}

\subsection{Cracking Process}
We now specify how an offline attacker would use the stolen hash to run a dictionary attack. Password distribution and the breakpoint distribution induce a joint distribution over pairs $(pw, t) \in \mathbb{P}\times \{t_1,\ldots, t_m\}$, we have $\Pr[(pw_i, t_j)] = \Pr[pw_i]q_j$ because of independence between label index and associated account. 

The adversary's strategy is to formulate a checking sequence $\pi =  \{(pw_i,t_j)\}$ with the purpose of finding the target $(pw_u,t_u)$.  In particular, the attacker will execute the first instruction in $\pi$, then second instruction if the first execution fails, etc,. An instruction $(pw_i,t_j)$ in $\pi$ means the adversary selects $pw_i$ as current guess and compute  the $j$th label for $pw_i$ i.e., evaluate $\MHF(pw_i, salt_u;t_j)$.  The cracking process terminates when the adversary found  the hidden target $(pw_u,t_u)$ or timeout. Note that the order of instructions in a checking sequence $\pi$ matters. A checking sequence is subject to \emph{legit restrictions}:
\begin{enumerate}
 \item \emph{Small label first}. If $(pw_i,t_{j_1})$ appears before $(pw_i,t_{j_2})$ in $\pi$, then it should be the case $t_{j_2}>t_{j_1}$.
 \item \emph{Label backward continuity}. If $(pw_i,t_j)\in \pi$ then $(pw_i, t_1),\ldots, (pw_i, t_{j-1})\in \pi$.
 \item \emph{No inversions}.  Inversions in the form of $(pw_{i_1}, t_{j_1})$, $(pw_{i_2}, t_{j_2})$, $(pw_{i_1}, t_{j_1}^{\prime})$ where $t_{j_1}^{\prime} > t_{j_1}$ are not allowed.
\end{enumerate} 
The first two restrictions state that the attacker cannot advance to  a larger label without computing all previous labels. The third is an assumption that we made, this assumption is valid because computing labels for $pw_{i_2}$ while storing labels for $pw_{i_1}$ will induce extra memory cost, which should be avoided.

\subsection{Attacker's Utility}
After specifying the restrictions for a legit checking sequence, we can formulate the the attacker's utility. 
Suppose the $k$th instruction in checking sequence $\pi$ is $\pi_k=(pw_i,t_j)$, then $\Pr[\pi_k] = \Pr[pw_i]\cdot q_j$. Define $\lambda(\pi, B) = \sum_{k=1}^B \Pr[\pi_k]$,  we write $\lambda(\pi,|\pi|)$ in short form  $\lambda(\pi)$.  Also we use $c(\pi_k)$ to denote the round cost of executing instruction $\pi_k$ where $c(\pi_k) = c_M(t_j^2 - t_{j-1}^2)$, specifically, $t_0 =0$ is introduced for convenience of notations. Then the attacker's utility is
\begin{equation}\label{eq:utility}\small
U_{adv} (v, \vec{q},\pi)= v\cdot\lambda(\pi) - \sum_{k = 1} ^ {|\pi|} c(\pi_k) \left(1- \lambda(\pi,k-1)\right).
\end{equation}
Given a checking sequence $\pi$ the attacker's success rate is $P_{adv} (\pi)=   \lambda(\pi)$, thus the gain (the first term of equation \eqref{eq:utility}) is $v\cdot\lambda(\pi)$.
The second term is the expected cost, the attacker only pay cost $c(\pi_k)$ if and only if all previous $k-1$ trials are in vain, which happens with probability $1-\lambda(\pi,k-1)$.

Besides legit restrictions that make a checking sequence valid a rational attacker would choose a checking sequence $\pi$ that satisfies \emph{opt restrictions}:
 \begin{enumerate}
 \item \emph{Popular password first}. If $(pw_{i_1},t_j)$ appears before $(pw_{i_2},t_j)$, then $\Pr[pw_{i_1}]>\Pr[pw_{i_2}]$.
  \item \emph{Password backward continuity}. If $(pw_i, t_j) \in \pi$ for some $j$, then $(pw_{i-1}, t_{j^{\prime}}) \in \pi$ for some $j^{\prime}$.
 \item \emph{Stop at $es$ boundary}. If $(pw_i, t_j)$  is the last instruction in $\pi$ where $pw_i \in es_k$, then $pw_{i+1} \in es_{k+1}$.
 \end{enumerate}

%\emph{Password forward continuity}. If $(pw_i, t_j) \in \pi^*$ for some $j$, then $(pw_{i+1}, t_j) \in \pi^*$ if $\Pr[pw_i] = \Pr[pw_{i+1}]$.

It can be easily proved that an attacker who violates opt restrictions will suffer utility loss.  Legit restrictions, together with the first 2 opt restrictions, determine a complete ordering, which we call \emph{natural ordering},  over all instructions $\{(pw_i, t_j)\}$, namely,
\begin{equation}\small
\begin{cases}
(pw_{i_1}, t_{j_1}) < (pw_{i_2}, t_{j_2}),\text{ if } \Pr[pw_{i_1}] > \Pr[pw_{i_2}],\\
(pw_i, t_{j_1}) < (pw_i, t_{j_2}),\text{ if } j_1 < j_2.
\end{cases}
\end{equation}

We use $\Pi(n,m)$ to denote the  sequence of all instructions for top $n$ passwords with respect to natural ordering,  
\begin{equation}\small
\Pi(n,m) \defeq (pw_1,t_1),\ldots,(pw_1,t_m),\ldots,(pw_n,t_1),\ldots,(pw_n,t_m).
\end{equation}
We say a sequence containing consecutive instructions  for a single password is a \emph{instruction bundle}, which is denoted by  
\begin{equation}
    \varpi_i(j_1,j_2) \defeq (pw_i,t_{j_1}),\ldots, (pw_i, t_{j_2}).
\end{equation}
Specifically,  $\varpi_i(j_1, j_2) = \emptyset$  when $j_1 = j_2 =0$.
Then the attacker's strategy $\pi$ is a sub-sequence of $\Pi(n_p,m)$ (recall that $n_p$ is the number of distinct passwords) in the form of 
\begin{equation}\small
\pi =  \oplus_{i^{\prime}=1}^{\len{\pi}} \varpi_{i^{\prime}}(1, \tau_{i^{\prime}}) 
\defeq \varpi_1(1, \tau_1)\circ \varpi_2(1, \tau_2)\circ \cdots \circ\varpi_{\lenz}(1, \tau_{\lenz}),
\end{equation}
where $\circ$ denotes the concatenation of two disjoint instruction sequence and $\len{\pi}$ is the largest index of password for which the attacker would check at least one label, which depends on the associated checking sequence, when the context is clear it is just written as $\lenz$.  Because of opt restriction 3, $\lenz $ can only take values in $\left\{0, |es_1|, |es_1|+|es_2|, \ldots,\sum_{k=1}^{n_e}|es_k|\right\}$.  Notice that $\pi$ is fully specified by the largest label index $\tau_i$ for $pw_i$.

\subsection{Stackelberg game}
We use Stackelberg game to model the interaction between the attacker and defender. 

stage 1: the server determines breakpoint distribution $\vec{q}$;

stage 2: The attacker's strategy is to select checking sequence $\pi$, or equivalently $\lenz $ and $\{\tau_i\}$.

Define server's utility to be $U_{ser}(v,\vec{q}) = - \lambda(\pi^*)$, where $\pi^*$ is the attacker's best response to defender's strategy $\vec{q}$ and password value $v$. At equilibrium no player has the incentive to deviate form her/his strategy, thus equilibrium profile $(\vec{q}^*, \pi^*)$ satisfies, 
\begin{equation}\small
\begin{cases}
 U_{adv}(v, \vec{q},\pi^*)  \geq U_{adv}(v, \vec{q},\pi),\, \forall  \pi,\\
 U_{ser}(v, \vec{q}^*) \geq  U_{ser}(v, \vec{q}), \, \forall \vec{q}.
\end{cases}
\end{equation} 

We use backward induction to find the equilibrium. First given password value $v$ and a concrete breakpoint distribution $\vec{q}$ we formulate the attacker's optimal strategy which maximize $U_{adv}(v, \vec{q},\pi)$ (Section \ref{sec:optimalstrategy}). Then taken the attacker's reaction into consideration the server determines $\vec{q}$ which minimizes attacker's success rate (Section \ref{sec:defender}).

\section{Attacker's Optimal Strategy}\label{sec:optimalstrategy}
In this section, we show how to compute the attacker's optimal strategy for both time-even breakpoints and cost-even breakpoints.
 Given  password value $v$ and label distribution $\vec{q}$, the attacker would choose $\pi^*=\argmax U_{adv}(v, \vec{q},\pi)$. 
 
 Before we introduce our algorithm used to find the optimal checking sequence, let us see why the native brute force algorithm is computationally infeasible. If the attacker chose to check top $\lenz$ passwords; for each password $pw_i$ the attacker has $m$ choices, namely, selects $\tau_i \in \{1,\ldots, m\}$. Thus the native brute force algorithm runs in time $\mathcal{O}\left(\sum_{\lenz= 1}^{n_p} m^{\lenz}\right)\subseteq \mathcal{O}(m^{n_p})$ with a very large exponent ($n_p\approx 2.14 \times 10^7$ for our largest dataset Linkedin, and $n_p\approx 3.74 \times 10^5$  for our smallest dataset Bfiled). This is why we need to design polynomial time algorithms.
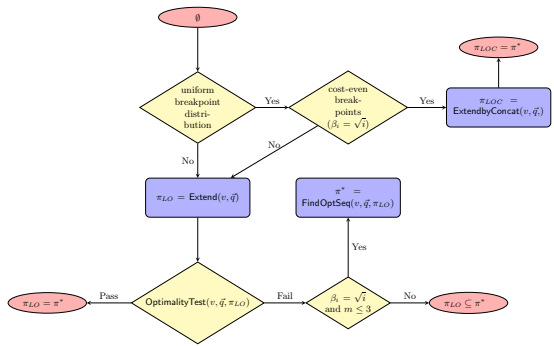
\begin{figure}[thb]
\centering
\scalebox{.4}{
\begin{tikzpicture}[node distance=2cm]
%\node (start) [startstop] {$\emptyset$};
%\node(pro1)[process, below of=start ]{$\EXC$};
%\node(dec1)[decision, below of=pro1, yshift=-0.5cm ]{if $\beta_i=\sqrt{i}$};

\tikzstyle{decision} = [diamond, draw, fill=yellow!30, text width=5em, minimum height = 3em, text centered, node distance=3cm, inner sep=0pt, aspect=1.62]
\tikzstyle{block} = [rectangle, draw, fill=blue!30, text width=10em, text centered, rounded corners, minimum height=4em]
\tikzstyle{line} = [draw, -latex']
\tikzstyle{cloud} = [draw, ellipse,fill=red!30, node distance=3cm, minimum height=1em]
% $\beta_i=\sqrt{i}$
% Place nodes
\node [cloud, text width=5em, text badly centered] (init) {$\emptyset$};
\node [decision, below of=init] (uniformdistribution) {uniform breakpoint distribution};
\node [block, text width=10em, below of=uniformdistribution, node distance=3cm] (extend) {$\pi_{LO} = \EX$};
\node [decision, right of=uniformdistribution, node distance = 5cm] (costeven) {cost-even breakpoints ($\beta_i=\sqrt{i}$)};
\node [block, text width=10em, right of=costeven, node distance=5cm] (concat) {$\pi_{LOC} = \EXC$};
\node [cloud, text width=5em, above of= concat, node distance= 2cm] (piLOC) {$\pi_{LOC} = \pi^* $};
\node [decision, text width = 12em,  below of=extend, node distance = 3.5cm] (OPT) { $\OT{\pi_{LO}} $};

\node [decision, right of=OPT, node distance = 5cm] (costeven2) {$\beta_i=\sqrt{i}$ and $m\leq 3$};

\node [cloud, text width=5em, left of= OPT, node distance= 5cm] (PiH) {$\pi_{LO}= \pi^* $};
\node [cloud, text width=5em, right of= costeven2, node distance= 4cm] (PiH-SubSetEq) {$\pi_{LO}\subseteq \pi^* $};
%\node [block, text width=10em, below of=costeven, node distance=3cm] (extend2) {$\pi_l = \EX$};
%\node [decision, text width = 10em, below of=extend2, node distance = 4cm] (OPT2) { $\OT{\pi_l} $};
%\node [cloud, text width=4em, right of= OPT2, node distance= 5cm] (PiH2) {$\pi_l= \pi^* $};
\node [block, text width=10em, below of=costeven, node distance=3cm] (FindGood) {$\pi^*  = \FG{\pi_{LO}}$};
%\node [cloud, text width=4em, right of= OPT2, node distance= 5cm] (PiH2) {$\pi_l= \pi^* $};

% \node [cloud, text width=7em, right of= FindGood, node distance= 5cm] (PiF) {$ \pi^* = \pi_{LO}^{j^*} + S_c^*(\pi_{LO}^{j^*}) $};

%\node [block, below of=init] (assign) {i = 1};
%\node [block, below of=assign] (increment) {increment i by 2};
%\node [block, left of=increment, node distance=3cm, text width=5em, fill=red!20] (incBy9) {increment i by 9};
%\node [block, below of=increment] (functionCall) {perform the function};
%\node [decision, below of=functionCall] (evenOddCheck) {is i even or odd};
%\node [cloud, below of=evenOddCheck] (stopping) {program closed then STOP};

% Draw edges
\draw [arrow] (init) -- (uniformdistribution);
\draw [arrow] (concat) -- (piLOC);
\draw [arrow] (extend) -- (OPT);
%\draw [arrow] (extend2) -- (OPT2);
% \draw [arrow] (FindGood) -- (PiF);
%\draw [arrow] (OPT) -- (PiH);
%\draw [arrow] (assign) -- (increment);
%\draw [arrow] (increment) -- (functionCall);
%\draw [arrow] (functionCall) -- (evenOddCheck);
\draw [arrow] (uniformdistribution) -- node [anchor=east] {No} (extend);
\draw [arrow] (uniformdistribution) -- node [anchor=south] {Yes} (costeven);
\draw [arrow] (costeven) -- node [anchor=south] {Yes} (concat);
\draw [arrow] (OPT) -- node[anchor=south]  {Pass} (PiH);
%\draw [arrow] (OPT2) -- node  {Pass} (PiH2);
%\draw [arrow] (OPT) -- node  {Fail} (PiH-SubSetEq);
\draw [arrow] (OPT) -- node[anchor=south] {Fail} (costeven2);
\draw [arrow] (costeven2)[anchor=south] -- node  {No} (PiH-SubSetEq);
\draw [arrow] (costeven2)[anchor=west] -- node  {Yes} (FindGood);
\draw [arrow] (costeven) -- node[anchor=south] {No} (extend);
%\draw [arrow] (costeven) -- node [near start] {No} (extend2);
%\draw [arrow] (incBy9) |- (assign);
%\draw [arrow] (evenOddCheck) |- (stopping);
%\draw [arrow] (uniformdistribution) -- ++(1.8cm,0cm) |- node [near start] {Yes} (init);
\end{tikzpicture}
}
\caption{Algorithm Flowchart}
\label{fig:alg}
\end{figure}
 
In the following subsections, we first specify a superset of $\pi^*$\footnote{We use the concept and notation of subset and superset for ordered sequences the way they were defined for regular set. If all elements of sequence $A$ are also elements of sequence $B$ regardless the order, we say $A \subseteq B$}, setting a boundary within which we will gradually extend the checking sequence from an empty one. Then we introduce our local search algorithm to find the optimal checking sequence (for most of the time). Our key intuition in designing algorithms is that an unchecked instruction bundle should be included into the optimal checking sequence if it provides non-negative marginal utility. Generally there are two local search directions, either concatenate instructions at the end of current checking sequence or insert instructions in the middle of current checking sequence.  After the local search algorithm terminates we reach a local optimum $\pi_{LO}$. Finally we design  algorithms to verify if the local optimum is also global optimum or promote the local optimum to global optimum under specifc parameter settings. As a overview we briefly summarize our results (also demonstrated in the flowchart, see Figure \ref{fig:alg}) in this section as follows:

\begin{itemize}
\item When we use cost-even breakpoints sampled from uniform distribution, namely, $\beta_i = \sqrt{i}$ and $q_i = \frac{1}{m}$, we have a local search algorithm $\EXC{\emptyset}$ which iteratively considers instruction bundle that can be concatenated, $\EXC{\emptyset}$ runs in time $\mathcal{O}(n_pm)$ and gives optimal checking sequence;
\item When breakpoints are cost-even ($\beta =\sqrt{i}$) but the distribution is non-uniform, we design an algorithm $\EX$ which returns a locally optimal checking sequence $\pi_{LO}$ in time $\mathcal{O}(n_pm)$. By locally optimal we mean that advancing any number of labels for any single password on the basis of $\pi_{LO}$ will decrease attacker's utility. 

After obtaining $\pi_{LO}$, we can run a polynomial algorithm $\OT{\pi_{LO}}$ to check if $\pi_{LO}$ is also a global optimum. If $\OT{\pi_{LO}}$ returns PASS, we know for sure that $\pi_{LO} = \pi^*$; otherwise, no conclusion can be drawn. If $m\leq 3$ we will use an efficient brute force algorithm $\FG{\pi_{LO}}$, which runs in time $\mathcal{O}(n_p^2)$, to the reach global optimum.

\item When $\beta \neq \sqrt{i}$, regardless of the breakpoint distribution we can still run $\EX$ to obtain locally optimal $\pi_{LO}$, and feed $\pi_{LO}$ to $\OT{\pi_{LO}}$. If $\OT{\pi_{LO}}$ returns PASS, again we have $\pi_{LO} = \pi^*$; if $\OT{\pi_{LO}}$ returns FAIL, we cannot deduce any information about the global optimality of $\pi_{LO}$; in this case, confirm that $\pi_{LO} = \pi^*$ or promote $\pi_{LO}$ to $\pi^*$ will take exponential time.

\end{itemize}

\subsection{Marginal Utility}
Since we are going to use marginal utility as metrics of state transition in local search, we first specify how to compute marginal utility.
% !TEX root = main.tex
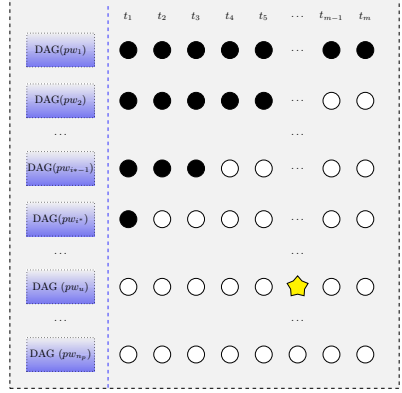
\begin{figure}[t]
\label{fig:pebbling}
\centering
\scalebox{0.45}{
\begin{tikzpicture}[node distance=2cm]
\centering
\tikzset{% This is the style settings for nodes
	blackC/.style={circle,minimum size=0.5cm,fill=black,draw=black},
	whiteC/.style={circle,minimum size=0.5cm,fill=white,draw=black},
	blueC/.style={circle,minimum size=0.5cm,fill=blue,draw=black},
	redC/.style={circle,minimum size=0.5cm,fill=red,draw=black},
	greenC/.style={circle,minimum size=0.5cm,fill=green,draw=black},
	violetC/.style={circle,minimum size=0.5cm,fill=red,draw=black},
	cli/.style={circle,minimum size=1cm,fill=white,draw},
	spl/.style={cli,append after command={
			node[circle,draw,dotted,
			minimum size=1.5cm] at (\tikzlastnode.center) {}}},
	c1/.style={-stealth,very thick,red!80!black},
	v2/.style={-stealth,very thick,yellow!65!black},
	v4/.style={-stealth,very thick,purple!70!black}}

\filldraw[dashed, top color=gray!10!,bottom color=gray!10!] (-1.5 ,-11) rectangle+(11.5,11.5);

\node at (2,0){$ t_1 $};
\node at (3,0){$ t_2 $};
\node at (4,0){$ t_3 $};
\node at (5,0){$ t_4 $};
\node at (6,0){$ t_5 $};
\node at (7,0){$ \ldots $};
\node at (8,0){$ t_{m-1} $};
\node at (9,0){$ t_m  $};

\filldraw[dotted, top color=white,bottom color=blue!50!] (-1 ,-1.5) rectangle+(2,1);
\node at (0,-1){DAG($ pw_1 $)};
\node[blackC] (0) at (2,-1){};
\node[blackC] (0) at (3,-1){};
\node[blackC] (0) at (4,-1){};
\node[blackC] (0) at (5,-1){};
\node[blackC] (0) at (6,-1){};
\node at (7,-1){$ \ldots $};
\node[blackC] (0) at (8,-1){};
\node[blackC] (0) at (9,-1){};

\filldraw[dotted, top color=white,bottom color=blue!50!] (-1 ,-3) rectangle+(2,1);
\node at (0,-2.5){DAG($ pw_2 $)};
\node[blackC] (0) at (2,-2.5){};
\node[blackC] (0) at (3,-2.5){};
\node[blackC] (0) at (4,-2.5){};
\node[blackC] (0) at (5,-2.5){};
\node[blackC] (0) at (6,-2.5){};
\node at (7,-2.5){$ \ldots $};
\node[whiteC] (0) at (8,-2.5){};
\node[whiteC] (0) at (9,-2.5){};

%%%%%
\node at (0,-3.5){$ \ldots $};
\node at (7,-3.5){$ \ldots $};

%%%%
\filldraw[dotted, top color=white,bottom color=blue!50!] (-1 ,-5) rectangle+(2,1);
\node at (0,-4.5){DAG($ pw_{i*-1} $)};
\node[blackC] (0) at (2,-4.5){};
\node[blackC] (0) at (3,-4.5){};
\node[blackC] (0) at (4,-4.5){};
\node[whiteC] (0) at (5,-4.5){};
\node[whiteC] (0) at (6,-4.5){};
\node at (7,-4.5){$ \ldots $};
\node[whiteC] (0) at (8,-4.5){};
\node[whiteC] (0) at (9,-4.5){};

\filldraw[dotted, top color=white,bottom color=blue!50!] (-1 ,-6.5) rectangle+(2,1);
\node at (0,-6){DAG($ pw_{i^*}$)};

\node[blackC] (0) at (2,-6){};
\node[whiteC] (0) at (3,-6){};
\node[whiteC] (0) at (4,-6){};
\node[whiteC] (0) at (5,-6){};
\node[whiteC] (0) at (6,-6){};
\node at (7,-6){$ \ldots $};
\node[whiteC] (0) at (8,-6){};
\node[whiteC] (0) at (9,-6){};

\node at (0,-7){$ \ldots $};
\node at (7,-7){$ \ldots $};

\filldraw[dotted, top color=white,bottom color=blue!50!] (-1 ,-8.5) rectangle+(2,1);
\node at (0,-8){DAG ($ pw_u$)};

\node[whiteC] (0) at (2,-8){};
\node[whiteC] (0) at (3,-8){};
\node[whiteC] (0) at (4,-8){};
\node[whiteC] (0) at (5,-8){};
\node[whiteC] (0) at (6,-8){};
\node[star,fill=yellow,minimum width=.65cm, draw = black] at (7,-8) {};
\node[whiteC] (0) at (8,-8){};
\node[whiteC] (0) at (9,-8){};

% \filldraw[dashed, draw =red,  top color=gray!30!,bottom color=gray!30!] (7 ,1) rectangle+(4,.5);
% \node at (9,1.2){$ pw_u$  with pepper $t_{m-1}$};
% \draw[thick, dashed, color = blue!80!] (10,-9)[->] -- +(-1.5,10);

\node at (0,-9){$ \ldots $};
\node at (7,-9){$ \ldots $};

\filldraw[dotted, top color=white,bottom color=blue!50!] (-1 ,-10.5) rectangle+(2,1);
\node at (0,-10){DAG ($ pw_{n_p}$)};

\node[whiteC] (0) at (2,-10){};
\node[whiteC] (0) at (3,-10){};
\node[whiteC] (0) at (4,-10){};
\node[whiteC] (0) at (5,-10){};
\node[whiteC] (0) at (6,-10){};
\node[whiteC] (0) at (7,-10){};
\node[whiteC] (0) at (8,-10){};
\node[whiteC] (0) at (9,-10){};

\draw[thick, dashed, color = blue!80!] (1.4,0.3)[-] -- +(0,-11.3);

\end{tikzpicture}
}
\caption{Password Cracking Process}
\medskip
\small
Black nodes denote current checking sequence $\pi$. White nodes denote unchecked instructions $\Pi(n_p,m)-\pi$. Star denotes unknown target $(pw_u, t_u)$.
\end{figure}

\begin{definition}
Fixing $v$ and $\vec{q}$, define $\Delta(\pi_1, \pi_2)$ to be marginal utility from strategy $\pi_1$ to $\pi_2$, namely,
\small
\begin{equation}
\Delta(\pi_1, \pi_2) \defeq U_{adv} (v, \vec{q},\pi_2) - U_{adv} (v, \vec{q},\pi_1).
\end{equation}
\normalsize
\end{definition}

For most of the time $\pi_2$ is the result of modifying $\pi_1$ which is called \emph{base},  in order to avoid redundantly repeating base  we often write $\con{e}{\pi_1}$ and $\ins{e}{\pi_1}$ to denote $\ma{\pi_1}{\pi_1 \circ e}$ and $\ma{\pi_1}{\pi_1 + e}$, respectively, where $e$ is some ordered set of instructions, referred to as \emph{extension}. Recall that $\circ$ is concatenation operation, here we formally introduce insertion operation $+$.
\begin{definition}
Given a checking sequence $\pi =\op{\lenz} $ and an instruction bundle $\varpi_{i^{\prime}}(j_1,j_2)$, define operation $\pi + \varpi_{i^\prime}(j_1,j_2)$ to be the checking sequence
\small
\begin{equation*}
\pi + \varpi_{i^\prime}(j_1,j_2) \defeq \op{i^{\prime}} \circ \varpi_{i^\prime}(j_1,j_2) \circ \oplus_{i = i^{\prime}+1}^{\lenz} \varpi_i(1, \tau_i).
\end{equation*}
\normalsize
\end{definition}

% Operation $\circ$ extends base at the tail while operation $+$ extends base from the body.
We discard superscript and  comprehensively write $\mar{e}{\pi}$ to denote the marginal utility by including $e$ into $\pi$, either through concatenation or insersion. Operations are valid only if the extension is \emph{compatible} with the base.  By compatible we mean the resulting checking sequence also satisfy both legit restrictions and opt restrictions. 

When $e$ is a singleton, from equation \eqref{eq:utility} we can derive the marginal utility by inserting instruction $e = (pw_i, t_j)\notin \pi$ to base $\pi$,
\small
\begin{equation}\label{eq:ins}
\ins{e}{\pi} = \Pr[pw_i]q_j \left(v + \sum_{e^{\prime} >e, e^{\prime}\in \pi} c(e^{\prime})\right) 
- \left(1-\sum_{e^{\prime} <e, e^{\prime}\in \pi} \Pr[e^{\prime}]\right) c_M(t_j^2 - t_{j-1}^2).
\end{equation}
\normalsize
where $\Pr[pw_i]q_j\sum_{e^{\prime} >e, e^{\prime}\in \pi} c(e^{\prime})$ is the influence of $e$ on future instructions since it eliminates some uncertainty about the user's password $pw_u$ thus reduces the \emph{expected} cost for future trials.

When $e$ is a singleton, marginal utility upon concatenation has no future influence, hence, 
\small
\begin{equation}\label{eq:con}
\con{e}{\pi} = \Pr[pw_i]q_j v - \left(1-\lambda(\pi)\right) c_M(t_j^2 - t_{j-1}^2).
\end{equation}
\normalsize

When $e$ consists of multiple consecutive instructions, the marginal utility can be computed by iteratively applying equation \eqref{eq:ins} and \eqref{eq:con}. Namely,
\small
\begin{equation}
	\mar{e}{\pi} = \sum_{i=1}^{|e|} \mar{e_i}{\pi \cup \{e_0,\ldots, e_{i-1}\}},
\end{equation}
\normalsize
where $e_0 = \emptyset$, $e_i$ is the $i$th instruction of $e$ and $\cup$  denotes inclusion (whether through concatenation or insertion) while maintaining natural ordering.

\subsection{A Superset of the Optimal Checking Sequence}
Before we present our algorithms we first show how to prune down the search space for $\pi^*$. Particularly, fixing $v$ and $\vec{q}$ we find an index $\lenmax $ such that $\pi^* \subseteq \Pi(\lenmax , m)$ i.e., $\pi^*$ will not even partially check passwords with rank larger than $\lenmax $. Thus there is no need to consider any instructions beyond $\Pi(\lenmax , m)$ in construction of the optimal checking sequence.
  
 \begin{restatable}{lemma}{resmarginal}
 \label{lemma:marginal}\small 
$
\con{\pi_3}{\pi_1} \leq \con{\pi_3}{\pi_2}, \, \text{ if } \lambda(\pi_1) \leq \lambda(\pi_2).
$\normalsize
 \end{restatable}
 
\begin{definition}
Fixing $v$ and $\vec{q}$ we define  
\begin{equation*}
\lenmax  \defeq 
\begin{cases}
  \max_i\{i:F(v, \vec{q}, i) \geq 0\}, \text{ if such } i \text{ exists}, \\
 0, \text{ o.w.}
\end{cases}
\end{equation*}
where 
\small\begin{equation*}
F(v,\vec{q},i) \defeq
\begin{cases}
 \max_{1\leq j \leq m}\{\margin{\emptyset}{\varpi_i(0,j)}\}, \text{ if }i=1, \\
 \max_{1\leq j \leq m}\{\con{\varpi_i(0,j)}{\Pi(i-1,m)}\}, \text{ o.w.}
\end{cases}
\end{equation*}
\end{definition}
Intuitively,
$\lenmax $ is the largest possible password index for which at least one of instruction bundles $\varpi_{\lenmax}(1, j), 1\leq j \leq m$ provide non-negative marginal utility no matter what previous instructions are. We remark even though there is no theoretical proof of monotonicity of $F(v, \vec{q},i)$, we have verified that $F(v, \vec{q},i)$ is decreasing in $i$ for our empirical password distribution. Note that by Lemma \ref{lemma:marginal}
we have
\small
\begin{equation*}
\con{\varpi_i(0,j)}{\op{i-1}} \leq F(v, \vec{q}, i),
\end{equation*}
\normalsize
if $F(v, \vec{q},i) < 0$, then $\varpi_i(0, j)$ would certainly provide negative marginal utility, thus cannot be included in $\pi^*$.  It is described in the following theorem.

\begin{restatable}{theorem}{restatesuperset}
\label{theorem:superset}
\small
\begin{equation*}
	\pi^* \subseteq \Pi(\lenmax, m).
\end{equation*}
\normalsize
\end{restatable}

\subsection{Extension by Concatenation}

We have established a superset of $\pi^*$ in last subsection, now we design a local search algorithm that gives us a checking sequence $\pi_{LOC}$ which is a subset of  $\pi^*$. Here, LOC stands for ``locally optimal with respect to concatenation." The sequence $\pi_{LOC}$ will be helpful to further prune down the search space for $\pi^*$. In fact, in the special case the breakpoint distribution is uniform ($q_i  = \frac{1}{m}$) and when we use cost-even breakpoints ($\beta
_i = \sqrt{i}$) we can prove that equality holds i.e.,  $\pi_{LOC}=\pi^*$ {\em is} the optimal solution. 

To find our sequence $\pi_{LOC}$ we start with the empty sequence of instructions and repeatedly include instructions that provide non-negative marginal utility upon concatenation to the current solution. We design a local search algorithm $\EXC{\emptyset}$ to find a checking sequence $\pi_{LOC}$. Our local search algorithm $\EXC{\emptyset}$ terminates after at most $n_p$ rounds.

 After the $i-1$th round we have $\pi_{LOC} \subseteq \Pi(i-1, m)$ i.e., the current solution only includes checking instructions for the first $i-1$ passwords. In the $i$th round we find an instruction bundle for password $i$ which maximizes (marginal) utility upon concatenation.  More specifically, in round $i$ we compute $\tau_i = \arg \max _{0\leq j\leq m} \{\con{\varpi_i(0,j)}{\pi_{LOC}}\}$ and append this instruction bundle to obtain an updated checking sequence $\pi_{LOC} = \pi_{LOC}\circ \varpi_i(0,\tau_i)$. Details can be found in Algorithm \ref{alg:EXbyConcat}. 

 \vspace{-0.5cm}
 \begin{algorithm}[h]\small
	\KwIn{$v$, $\vec{q}$}
	\KwOut{$\pi_{LOC}$}
	\BlankLine
	$\pi_{LOC} = \pi$\;
	$start = i^*(\pi_{LOC})$\;
	\For{$i = start:n_p$}
	{
		\For{$j = 0:m$}
		{
			Compute $\con{\varpi_i(0,j)}{\pi_{LOC}}$\;
		}
		$\tau_i = \arg \max _{0\leq j\leq m} \{\con{\varpi_i(0,j)}{\pi_{LOC}}\}$\;
		\If{$\tau_i > 0$}
		{
			$\pi_{LOC} = \pi_{LOC} \circ \vp{i}$\;
			\lElse{break}
		}
	}
	\KwRet{$\pi_{LOC}$}
	\caption{$\EXC{\pi}$}
	\label{alg:EXbyConcat}
 \end{algorithm}
 \vspace{-0.5cm}

We can use equation \eqref{eq:con} to compute the marginal utility in time $\mathcal{O}(1)$ by caching previously computed values of $\lambda(\pi)$.
Thus, $\EXC{\emptyset}$ runs in time $\mathcal{O}(\lenmax m) \subseteq \mathcal{O}(n_pm)$, recall that $n_p$ is the number of distinct password.

\begin{restatable}{theorem}{restateconcat}
	\label{theorem:concat}
\small
\begin{equation*}
	\pi_{LOC}\subseteq \pi^*.
\end{equation*}
\normalsize
\end{restatable}

From Theorem \ref{theorem:superset} and Theorem \ref{theorem:concat}, it is easy to derive the following corollaries.

\begin{corollary}
\small
$$\len{\pi_{LOC}}\leq \len{\pi^*}\leq  \lenmax ,$$
\normalsize
and
\small
$$\len{\pi_{LOC}}, \len{\pi^*}, \lenmax \in \left\{x_0,x_1,\ldots,x_{n_e}\right\}, $$
\normalsize
where
\small 
\begin{equation}
x_k =
\begin{cases}
0, \text{ if } k =0,\\
\sum_{k^{\prime} = 1}^k |es_{k^{\prime}}|, \text{ if } k =  1,\ldots, n_e.\\
\end{cases}
\end{equation}\normalsize
\end{corollary}

\begin{corollary}
\small$$\lambda(\pi_{LOC})\leq P_{adv} = \lambda(\pi^*) \leq \lambda\left(\Pi(\lenmax , m)\right).$$\normalsize
\end{corollary}

Now we have a polynomial algorithm that returns a checking sequence $\pi_{LOC}$ locally optimal with respect to concatenation. The following theorem states that $\pi_{LOC} = \pi^*$ if breakpoints are cost-even and follow uniform distribution.

\begin{restatable}{theorem}{restatecopt}
	\label{theorem:copt}
When $q_i = \frac{1}{m}$ and $\beta_i = \sqrt{i}$, $\EXC{\emptyset}$ returns the optimal checking sequence, i.e., $\pi_{LOC} = \pi^*$.
\end{restatable}

Even though the attacker behaviors optimally---following strategy $\pi^*$. We can guarantee that our mechanism results in lower (or equal if no passwords are cracked) percentage of cracked passwords than deterministic cost hashing, which is captured by Theorem \ref{theorem:costeven}.
\begin{restatable}{theorem}{restatecosteven}
\label{theorem:costeven}
When $\beta_i = \sqrt{i}$ and $q_i = \frac{1}{m}$ then,
\small$\lambda(\pi^*) \leq P_{adv}^{d},$\normalsize
where $P_{adv}^d$ is the percentage of cracked passwords in traditional deterministic cost hashing.
\end{restatable}

We have shown that our mechanism configured with cost-even breakpoints sampled from uniform distribution will only decrease the percentage of cracked passwords.
In the next subsections we consider how the attacker would react to general configuration of the mechanism. 

\subsection{Local Search in Two Directions}
In the previous section we introduced an algorithm $\EXC{\emptyset}$ to produce a locally optimal solution $\pi_{LOC}$  with respect to concatenation. We showed the instruction sequence $\pi_{LOC}$ is a subset of the instructions in $\pi^*$ and argued that in specific cases the algorithm is guaranteed to find the optimal solution. However, in more general cases the local optimum may not be globally optimum. One possible reason for this is that there may be a missing instruction from $\pi^*$ that we would like to insert into the middle of the checking sequence $\pi_{LOC}$, while our local search algorithm only considers instructions that can be appended to $\pi_{LOC}$.  

In this subsection we extend the local search algorithm to additionally consider insertions. Note that we can still use local search to test if $\ins{\bundle}{\pi}\geq 0 $ inserting instruction bundle $\bundle$ improves the overall utility. We define an algorithm $\EXI{\pi}$ which performs such an update. Combining $\EXC{\pi}$ and $\EXI{\pi}$,
we design an Algorithm $\EX$ to construct a checking sequence $\pi_{LO}$ (LO=Locally Optimal) which is locally optimal with respect to both operations: concatenation and insertions. Specifically, after each call of $\EXI{\pi}$ we immediately run $\EXC{\pi}$ to ensure that the solution is still locally optimal with respect to concatenation. See Algorithm \ref{algo:EX} for details. The algorithms still maintains the invariant that $\pi_{LO}$ is a subset of $\pi^*$ --- see Theorem \ref{theorem:EX}.

Given $\pi_{LOC}$ computed in time $\mathcal{O}(n_pm)$,  the number of unchecked instructions is upper bounded by $|\Pi(\lenmax , m)| - |\pi_{LOC}|$. By caching the probability summation of previous and future instructions at each insertion position, verify if an instruction bundle is profitable and update the checking sequence take time $\mathcal{O}(1)$. One pass of \texttt{repeat} loop of Algorithm \ref{algo:EX} takes time $\mathcal{O}(|\Pi(\lenmax , m)| - |\pi_{LOC}|) \subseteq \mathcal{O}(n_pm)$, the execution time of \texttt{repeat} loop execution is finite (in experiment it terminates after at most 3 passes). Therefore, $\EX$ runs in time $\mathcal{O}(n_pm)$.

\vspace{-0.5cm}
\begin{algorithm}[h]\small
	\KwIn{$v$, $\vec{q}$, $\pi$}
	\KwOut{$\pi_{LOI}$}
	\BlankLine
	$\pi_{LOI} = \pi$\;
	\While{$e$ exists such that $\ins{e}{\pi_{LOI}}\geq 0$}
	{
		$\pi_{LOI} = \pi_{LOI}+ e$
	}
	\KwRet{$\pi_{LOI}$}
	\caption{$\EXI{\pi}$}
	\label{algo:instert}
\end{algorithm}

% \begin{algorithm}[htb]\small
% % \caption{$\EXI{\pi}$}
% % \label{algo:instert}
% \begin{algorithmic}[1]
% \REQUIRE{$v$, $\vec{q}$, $\pi$}
% \ENSURE{$\pi_{LOI}$}
% \STATE $\pi_{LOI} = \pi$
% \WHILE{$e$ exists such that $\ins{e}{\pi_{LOI}}\geq 0$}
% \STATE $\pi_{LOI} = \pi_{LOI}+ e$
% \ENDWHILE
% \RETURN $\pi_{LOI}$
% \end{algorithmic}
% \end{algorithm}
\vspace{-1cm}

\begin{algorithm}[h]\small
	\KwIn{$v$, $\vec{q}$}
	\KwOut{$\pi_{LO}$}
	\BlankLine
	$\pi_{LO} = \EXC{\emptyset}$\;
	\Repeat{no single profitable instruction bundle exist}
	{
		$\pi_{LO} = \EXI{\pi_{LO}}$\;
		$\pi_{LO} = \EXC{\pi_{LO}}$\;
	}
	\KwRet{$\pi_{LO}$}
	\caption{$\EX$}
	\label{algo:EX}
\end{algorithm}
\vspace{-1cm}

% \begin{algorithm}[htb]\small
% % \caption{$\EX$}
% % \label{algo:EX}
% \begin{algorithmic}[1]
% \REQUIRE{$v$, $\vec{q}$}
% \ENSURE{$\pi_{LO}$}
% \STATE $\pi_{LO} = \EXC{\emptyset}$
% \REPEAT
% \STATE $\pi_{LO} = \EXI{\pi_{LO}}$
% \STATE $\pi_{LO} = \EXC{\pi_{LO}}$
% \UNTIL{no single profitable instruction bundle exist}
% \RETURN $\pi_{LO}$
% \end{algorithmic}
% \end{algorithm}

\begin{restatable}{lemma}{resplus}
	\label{lemma:plus}
If $\pi\subseteq \pi^*$ and $\ins{e}{\pi}\geq 0$ then $\pi + e \subseteq \pi^* $.
\end{restatable}

Lemma \ref{lemma:plus} guarantees that $+$ operation preserves the invariance that our construction is subset of $\pi^*$. Naturally follows Theorem \ref{theorem:EX}, which states the output of $\EX$ is a subset of $\pi^*$.

\begin{restatable}{theorem}{restateex}
	\label{theorem:EX}
Let $\pi_{LO} = \EX$, then
\small$\pi_{LO} \subseteq \pi^*.$\normalsize
\end{restatable}

Since we are using local search to construct $\pi_{LO}$, together with Theorem \ref{theorem:EX} we know $\pi_{LO}$ is a local optimum. When Algorithm \ref{algo:EX} terminates,   advancing any number of labels for any single password cannot improve the overall utility, but there is no guarantee of utility reduction upon inclusion of multiple instruction bundles that associated with different passwords. In the next subsection we will discuss how to verify if the local optimum $\pi_{LO}$ is indeed the global optimum and design an efficient brute force algorithm that improves local optimum to global optimum under specific parameter settings.

\subsection{Optimality Test and Globally Optimal Checking Sequence}

In the previous subsections, we designed a polynomial algorithm $\EX$ to construct locally optimal checking sequence $\pi_{LO}$ with respect to insertions and concatenation. We also proved that the sequence $\pi_{LO}$ is a subset of the optimal sequence $\pi^*$. In practice we find that it is often the case that $\pi_{LO}=\pi^*$ and we give an efficient heuristic algorithm which (often) allows us to confirm the global optimality of $\pi_{LO}$. In particular, our procedure will never falsely indicate that $\pi_{LO}=\pi^*$ though it may occasionally fail to confirm that this is the case. When our optimality test fails, we design algorithms to promote locally optimal solution to globally optimal solution for cost-even breakpoints and $m\leq 3$, see full version of this paper for details.

\section{Defender's Optimal Strategy}\label{sec:defender}
When making decisions about breakpoint distribution, the defender will  take attacker's best response into consideration. Specifically, the defender would choose $\vec{q}^* = \arg \min \lambda(\pi^*)$ where $\pi^*=\arg\max U_{adv}(v, \vec{q},\pi))$. Formally, the optimization problem ($\mathsf{OPT}$) is
\small
\begin{mini}|s|
{\vec{q}}{\lambda(\pi^*) }
{}{}
\addConstraint{0 \leq q_i \leq 1, \; \forall 1 \leq i\leq m}
\addConstraint{\sum_{i=1}^mq_i = 1}
\addConstraint{\sum_{i=1}^mq_i c_Mt_i^2 \leq C_{max}}
\end{mini}
\normalsize
The optimization goal is to minimize attacker's success rate, the first two constrains guarantee $q_i$ are valid probabilities, the third constraint guarantees that the expected cost does not exceed maximum workload $C_{max}$.

We use black box optimization solvers to address $\mathsf{OPT}$ since there is no closed form expression of $\lambda(\pi^*)$. Notice that there are only $m-1$ independent variables, we select $q_2,\ldots,q_m$ to be optimization variables to further prune down the search space. Define $\alpha = \frac{C_{max}}{c_mt_1^2}$, then the workload constraint becomes $\sum_{i=1}^m \beta_i^2q_i\leq \alpha$, together with probability constraint, we have $q_i\leq \min\left\{1, \frac{\alpha-1}{\beta_i^2 -1}\right\}$. For fixed $v$, define $f$ to be the function mapping $\vec{q}$ to $P_{adv}$. Therefore $\mathsf{OPT}$ is reduced to $\mathsf{OPT}^{\prime}$:
\small
\begin{mini}|s|
{q_2,\ldots,q_m}{ f + \text{penalty}_1(q_2,\ldots,q_m) + \text{penalty}_2(q_2,\ldots,q_m)}
{}{}
\end{mini}
\normalsize
where 
\begin{equation*}\small
\text{penalty}_1(q_2,\ldots,q_m) =
\begin{cases}  0, \text{ if } \sum_{i=2}^m (\beta_i^2-1)q_i \leq \alpha - 1, \\ 
 cons_1 + cons_2 * \sum_{i=2}^m (\beta_i^2-1)q_i, \text{ ow},\\
\end{cases}
\end{equation*}
and 
\begin{equation*}\small
\text{penalty}_2(q_2,\ldots,q_m) =
\begin{cases}  0, \text{ if } \sum_{i=2}^m q_i \leq 1, \\ 
 cons_3 + cons_4 * \sum_{i=2}^m q_i, \text{ ow}.\\
\end{cases}
\end{equation*}
$cons_i$ are constants larger than 1. 
We refer to the black box solver as $\OPTDIS$, detailed discussion about $OPTDIS$ can be found in the full version of this paper. 

\section{Experiments}
\subsection{Experiment Setup}
In this section, we design experiments to validate our mechanism. In particular, we consider 8 datasets (Bfield, Brazzers, Clixsense, CSDN, Linkedin, Neopets, Rockyou, Webhost) with various size and they are given in format of equivalence sets $D_{es} = \{(f_i,s_i)\}$. We only report results of top 2 largest datasets (in terms of $n_p$): Linkedin and Neopets due to length limit and move results of other 6 datases to full version of this paper. For each dataset we plot attacker's success rate $\lambda(\pi^*)$ (there are a few exceptions where $\OT{\pi}$ fails thus $\pi^*$ is not available) against $v/C_{max}$, where $\pi^*$ is computed by algorithms specified in the flow chart. Since the attacker is playing the same Stackelberg game against all accounts, success rate of cracking one account can also be interpreted as percentage of cracked passwords in that dataset.

In Figure \ref{fig:timeeven} we consider time-even breakpoints with uniform distribution. We are essentially plotting $\lambda(\pi_{LO})$ against $C_{max}$ and have $\pi_{LO} = \pi^*$ for most of the time; if $\pi^*$ fails $\OT{\pi_{LO}}$, the corresponding point is marked with $\times$. Figure \ref{fig:costeven} shows the results under cost-even breakpoints and uniform distribution. In Figure \ref{fig:costopt}, we fix $m=3$ and breakpoints being cost-even, and run our algorithm $\OPTDIS$ whose optimization solver is implemented with BITEOPT~\cite{biteopt2021}, to find the optimal breakpoint distribution.

In experiments we use empirical distribution $\mathcal{D}_e$ (namely, $\Pr_{pw\sim \mathcal{D}_e}[pw] = f_i/n_a$ where $f_i$ is the frequency of $pw$) to approxmiate password distribution $\mathcal{P}$. The drawback is that the tail of empirical distribution $\mathcal{D}_e$ would significantly diverge from real distribution $\mathcal{P}$. We follow the approach of \cite{DAHash} and use Good-Turing Frequency estimation to upbound the CDF divergence $E$ between $\mathcal{D}_e$ and $\mathcal{P}$. In particular, we use yellow (resp. red) to denote the unconfident region where $E>0.01$ (resp. $E>0.1$). We only report results within the confident region.

% !TEX root = main.tex
\tikzDefaultsTime
\begin{figure*}[ht]\centering

\subfloat[Linkedin]{
\begin{tikzpicture}[scale=0.7, spy using outlines=
	{circle, magnification=15, connect spies}]
\begin{semilogxaxis}
% benchmark
\addplot file {./newfig/sr/benchmark/linkedin_benchmark.dat};

\addplot file {./newfig/sr/timeeven_m3/linkedin_timeeven_m3.dat};

\addplot+[only marks] coordinates
{
(4e+07,0.602737)
(5e+07,0.671147)
(6e+07,0.814533)
};

\path[name path=A] (axis cs:(3* 1e7,0) -- (axis cs:(3* 1e7,1);
\path[name path=B] (axis cs:(9 * 1e7,0) -- (axis cs:(9 * 1e7,1);
\tikzfillbetween[of=A and B, on layer=main]{red, opacity=0.2};

\path[name path=A1] (axis cs:(2* 1e7,0) -- (axis cs:(2* 1e7,1);
\tikzfillbetween[of=A1 and B, on layer=main]{yellow, opacity=0.2};

\node at (axis cs:5*1e7,0.5) {\rotatebox{90}{\tiny uncertain region}};

\coordinate (spypoint1) at (axis cs: 25, -0.1);
\coordinate (magnifyglass1) at (axis cs: 300, 0.3);

\end{semilogxaxis}

\spy [blue, size=1 cm] on (spypoint1) in node[fill=white] at (magnifyglass1);

\end{tikzpicture}
\label{fig:linkedintime}
}
\hfill
\subfloat[Neopets]{
\begin{tikzpicture}[scale=0.7, spy using outlines=
	{circle, magnification=15, connect spies}]
\begin{semilogxaxis}
%deterministic
\addplot file {./newfig/sr/benchmark/neopets_benchmark.dat};

\addplot file {./newfig/sr/timeeven_m3/neopets_timeeven_m3.dat};

\addplot+[only marks] coordinates
{
(5e+06,0.390334)
(2e+07,0.711407)
(3e+07,0.759299)
};

\path[name path=A] (axis cs:(2* 1e7,0) -- (axis cs:(2* 1e7,1);
\path[name path=B] (axis cs:(9 * 1e7,0) -- (axis cs:(9 * 1e7,1);
\tikzfillbetween[of=A and B, on layer=main]{red, opacity=0.2};

\path[name path=A1] (axis cs:(5* 1e6,0) -- (axis cs:(5* 1e6,1);
\tikzfillbetween[of=A1 and B, on layer=main]{yellow, opacity=0.2};

\node at (axis cs:5*1e7,0.5) {\rotatebox{90}{\tiny uncertain region}};

\coordinate (spypoint) at (axis cs: 390, -0.095);
\coordinate (magnifyglass) at (axis cs: 300, 0.3);
\end{semilogxaxis}
\spy [blue, size=1 cm] on (spypoint)
   in node[fill=white] at (magnifyglass);
\end{tikzpicture}
\label{fig:neopetstime}
}
% \hfill
% \subfloat[Yahoo]{
% \begin{tikzpicture}[scale=0.7, spy using outlines=
% 	{circle, magnification=10, connect spies}]
% \begin{semilogxaxis}
% % benchmark
% \addplot file {./newfig/sr/benchmark/yahoo_benchmark.dat};

% \addplot file {./newfig/sr/timeeven_m3/yahoo_timeeven_m3.dat};

% \addplot+[only marks] coordinates
% {
% (3e+07,0.694921)
% (4e+07,0.694921)
% };

% \path[name path=A] (axis cs:(2* 1e7,0) -- (axis cs:(2* 1e7,1);
% \path[name path=B] (axis cs:(9 * 1e7,0) -- (axis cs:(9 * 1e7,1);
% \tikzfillbetween[of=A and B, on layer=main]{red, opacity=0.2};

% \path[name path=A1] (axis cs:(7* 1e6,0) -- (axis cs:(7* 1e6,1);
% \tikzfillbetween[of=A1 and B, on layer=main]{yellow, opacity=0.2};

% \node at (axis cs:5*1e7,0.5) {\rotatebox{90}{\tiny uncertain region}};

% \coordinate (spypoint) at (axis cs: 54, -0.123);
% \coordinate (magnifyglass) at (axis cs: 300, 0.2);
% \end{semilogxaxis}
% \spy [blue, size=1 cm] on (spypoint)
%    in node[fill=white] at (magnifyglass);
% \end{tikzpicture}
% \label{fig:Yahootime}
% }
\caption{Time-Even Breakpoints, Uniform Breakpoint Distribution} 
\label{fig:timeeven}
\vspace{-0.5cm}
\end{figure*}
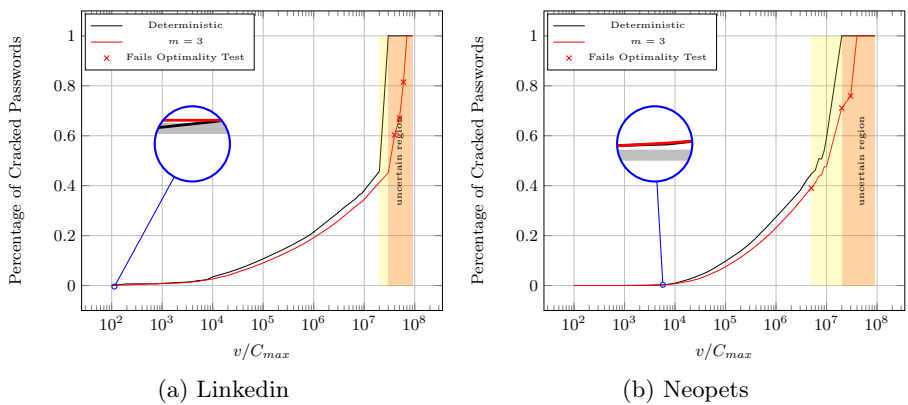
% !TEX root = main.tex
\tikzDefaultsCost
\begin{figure*}[ht]
\centering

\subfloat[Linkedin]{
\begin{tikzpicture}[scale=0.7]
\begin{semilogxaxis}
% benchmark
\addplot+[stack plots=y] file {./newfig/sr/benchmark/linkedin_benchmark.dat};
%costeven_m2
\addplot file {./newfig/sr/costeven_m2/linkedin_costeven_m2.dat};

\addplot file {./newfig/sr/costeven_m3/linkedin_costeven_m3.dat};

\addplot file {./newfig/sr/costeven_m7/linkedin_costeven_m7.dat};

\addplot file {./newfig/sr/costeven_m99/linkedin_costeven_m99.dat};

\addplot+[stack plots=y, stack dir=minus] file {./newfig/sr/costeven_m99/linkedin_costeven_m99.dat};

\path[name path=A] (axis cs:(3* 1e7,0) -- (axis cs:(3* 1e7,1);
\path[name path=B] (axis cs:(9 * 1e7,0) -- (axis cs:(9 * 1e7,1);
\tikzfillbetween[of=A and B, on layer=main]{red, opacity=0.2};

\path[name path=A1] (axis cs:(2* 1e7,0) -- (axis cs:(2* 1e7,1);
\tikzfillbetween[of=A1 and B, on layer=main]{yellow, opacity=0.2};

\node at (axis cs:5*1e7,0.5) {\rotatebox{90}{\tiny uncertain region}};

\end{semilogxaxis}
\end{tikzpicture}
\label{fig:linkedincost}
}
\hfill
\subfloat[Neopets]{
\begin{tikzpicture}[scale=0.7]
\begin{semilogxaxis}
%deterministic
\addplot+[stack plots=y] file {./newfig/sr/benchmark/neopets_benchmark.dat};
%costeven_m2
\addplot file {./newfig/sr/costeven_m2/neopets_costeven_m2.dat};

\addplot file {./newfig/sr/costeven_m3/neopets_costeven_m3.dat};

\addplot file {./newfig/sr/costeven_m7/neopets_costeven_m7.dat};

\addplot file {./newfig/sr/costeven_m99/neopets_costeven_m99.dat};

\addplot+[stack plots=y, stack dir=minus] file {./newfig/sr/costeven_m99/neopets_costeven_m99.dat};

\path[name path=A] (axis cs:(2* 1e7,0) -- (axis cs:(2* 1e7,1);
\path[name path=B] (axis cs:(9 * 1e7,0) -- (axis cs:(9 * 1e7,1);
\tikzfillbetween[of=A and B, on layer=main]{red, opacity=0.2};

\path[name path=A1] (axis cs:(5* 1e6,0) -- (axis cs:(5* 1e6,1);
\tikzfillbetween[of=A1 and B, on layer=main]{yellow, opacity=0.2};

\node at (axis cs:5*1e7,0.5) {\rotatebox{90}{\tiny uncertain region}};

\end{semilogxaxis}
\end{tikzpicture}
\label{fig:neopetscost}
}
% \hfill
% \subfloat[Yahoo]{
% \begin{tikzpicture}[scale=0.7]
% \begin{semilogxaxis}
% % benchmark
% \addplot file {./newfig/sr/benchmark/yahoo_benchmark.dat};

% %costeven_m2
% \addplot+[stack plots=y] file {./newfig/sr/costeven_m2/yahoo_costeven_m2.dat};

% \addplot file {./newfig/sr/costeven_m3/yahoo_costeven_m3.dat};

% \addplot file {./newfig/sr/costeven_m7/yahoo_costeven_m7.dat};

% \addplot file {./newfig/sr/costeven_m99/yahoo_costeven_m99.dat};

% \addplot+[stack plots=y, stack dir=minus] file {./newfig/sr/costeven_m99/yahoo_costeven_m99.dat};

% \path[name path=A] (axis cs:(2* 1e7,0) -- (axis cs:(2* 1e7,1);
% \path[name path=B] (axis cs:(9 * 1e7,0) -- (axis cs:(9 * 1e7,1);
% \tikzfillbetween[of=A and B, on layer=main]{red, opacity=0.2};

% \path[name path=A1] (axis cs:(7* 1e6,0) -- (axis cs:(7* 1e6,1);
% \tikzfillbetween[of=A1 and B, on layer=main]{yellow, opacity=0.2};

% \node at (axis cs:5*1e7,0.5) {\rotatebox{90}{\tiny uncertain region}};

% \end{semilogxaxis}
% \end{tikzpicture}
% \label{fig:yahoocost}
% }
\caption{Cost-Even Breakpoints, Uniform Breakpoint Distribution} 
\label{fig:costeven}
\vspace{-0.5cm}
\end{figure*}
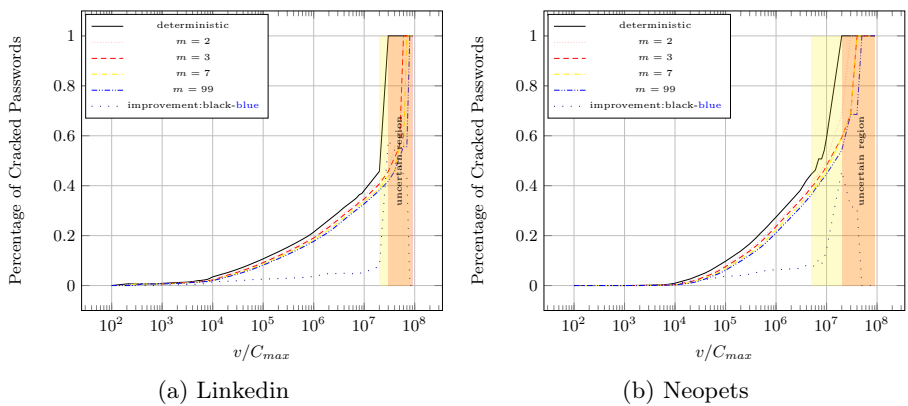
% !TEX root = main.tex
\tikzDefaultsBite
\begin{figure*}[ht]
\centering
\subfloat[Linkedin]{
\begin{tikzpicture}[scale=0.7]
\begin{semilogxaxis}
% benchmark
\addplot file {./newfig/sr/benchmark/linkedin_benchmark.dat};

\addplot file {./newfig/sr/costeven_m3/linkedin_costeven_m3.dat};

\addplot file {./newfig/sr/biteopt_m3/linkedin_biteopt_m3.dat};

% \addplot file {./newfig/sr/costeven_m7/linkedin_costeven_m7.dat};

\path[name path=A] (axis cs:(3* 1e7,0) -- (axis cs:(3* 1e7,1);
\path[name path=B] (axis cs:(9 * 1e7,0) -- (axis cs:(9 * 1e7,1);
\tikzfillbetween[of=A and B, on layer=main]{red, opacity=0.2};

\path[name path=A1] (axis cs:(2* 1e7,0) -- (axis cs:(2* 1e7,1);
\tikzfillbetween[of=A1 and B, on layer=main]{yellow, opacity=0.2};

\node at (axis cs:5*1e7,0.5) {\rotatebox{90}{\tiny uncertain region}};

\end{semilogxaxis}
\end{tikzpicture}
\label{fig:linkedinopt}
}
\hfill
\subfloat[Neopets]{
\begin{tikzpicture}[scale=0.7]
\begin{semilogxaxis}
%deterministic
\addplot file {./newfig/sr/benchmark/neopets_benchmark.dat};

\addplot file {./newfig/sr/costeven_m3/neopets_costeven_m3.dat};

\addplot file {./newfig/sr/biteopt_m3/neopets_biteopt_m3.dat};
%costeven_m7
% \addplot file {./newfig/sr/costeven_m7/neopets_costeven_m7.dat};

% \path[name path=A] (axis cs:(2* 1e7,0) -- (axis cs:(2* 1e7,1);
\path[name path=B] (axis cs:(9 * 1e7,0) -- (axis cs:(9 * 1e7,1);
\tikzfillbetween[of=A and B, on layer=main]{red, opacity=0.2};

\path[name path=A1] (axis cs:(5* 1e6,0) -- (axis cs:(5* 1e6,1);
\tikzfillbetween[of=A1 and B, on layer=main]{yellow, opacity=0.2};

\node at (axis cs:5*1e7,0.5) {\rotatebox{90}{\tiny uncertain region}};

\end{semilogxaxis}
\end{tikzpicture}
\label{fig:neopetsopt}
}
% \hfill
% \subfloat[Yahoo]{
% \begin{tikzpicture}[scale=0.7]
% \begin{semilogxaxis}
% % benchmark
% \addplot file {./newfig/sr/benchmark/yahoo_benchmark.dat};

% \addplot file {./newfig/sr/costeven_m3/yahoo_costeven_m3.dat};

% \addplot file {./newfig/sr/biteopt_m3/yahoo_biteopt_m3.dat};
% %costeven_m7
% % \addplot file {./newfig/sr/costeven_m7/yahoo_costeven_m7.dat};

% \path[name path=A] (axis cs:(2* 1e7,0) -- (axis cs:(2* 1e7,1);
% \path[name path=B] (axis cs:(9 * 1e7,0) -- (axis cs:(9 * 1e7,1);
% \tikzfillbetween[of=A and B, on layer=main]{red, opacity=0.2};

% \path[name path=A1] (axis cs:(7* 1e6,0) -- (axis cs:(7* 1e6,1);
% \tikzfillbetween[of=A1 and B, on layer=main]{yellow, opacity=0.2};

% \node at (axis cs:5*1e7,0.5) {\rotatebox{90}{\tiny uncertain region}};

% \node at (axis cs:5*1e7,0.5) {\rotatebox{90}{\tiny uncertain region}};

% \end{semilogxaxis}
% \end{tikzpicture}
% \label{fig:yahooopt}
% }
\caption{Cost-Even Breakpoints, Optimized Breakpoint Distribution} 
\label{fig:costopt}
\vspace{-0.5cm}
\end{figure*}
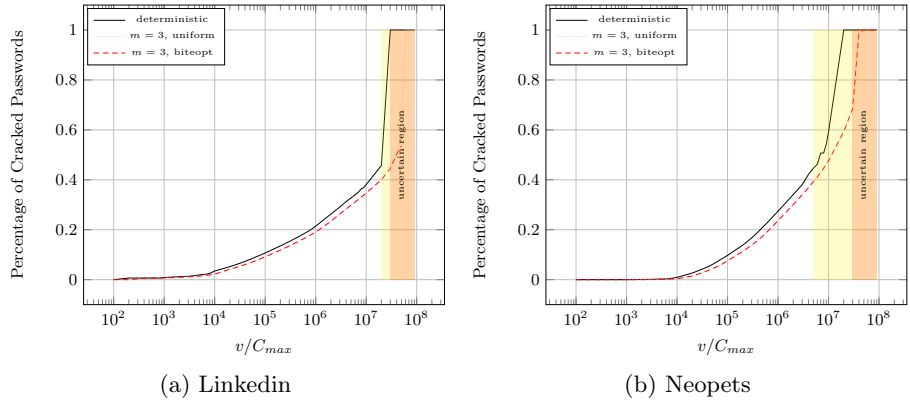

\subsection{Experiment Analysis}

\paragraph{Time-Even Breakpoints and Uniform Distribution.}
First we would like to demonstrate time-even breakpoints are not suitable. In Figure \ref{fig:timeeven}, we show attacker's success rate against $v/C_{max}$ when the attacker plays locally optimal strategy $\pi_{LO}$ and breakpoint distribution is uniform, namely $\beta_i=i$ and $q_i = \frac{1}{m}, \forall i\leq m$. For most $v/C_{max}$ ratios we have $\pi_{LO}=\pi^*$, only a few points fails $\OT{\pi_{LO}}$ and they are in unconfident region. Since the attacker is playing the same game against every account, success rate can be interpreted as the percentage of cracked accounts.

From Figure \ref{fig:timeeven}, we observe that time-even breakpoints can reduce attacker's success rate for most of the time. Unfortunately, this is not always the case. Take Linkedin as example, when $v/C_{max} = 100$, no account is cracked in deterministic cost hashing. However, there are 0.2\% accounts that would be cracked if time-even breakpoints were used. Similar phenomenon can be observed in other datasets, which is amplified by circles in Figure \ref{fig:timeeven}. When we randomize the label index in password hashing with MHF, we have $\frac{\Pr[pw_u]\Pr[t_1]} {\Pr[pw_u]} = 1/m$ while $\frac{c_Mt_1^2}{C_{max}} = \frac{1}{(m+1)(2m+1)} \in \mathcal{O}(\frac{1}{m^2})$. In other words, the probability that the attacker successfully guesses the first label drops lineally while the cost of making that guess drops quadratically. This is the reason that time-even breakpoints might have negative effect at the beginning of these plots. 

Here we also give a contrived example to show that time-even breakpoints could be harmful. Suppose a dataset has 2 passwords, each occurs with probability $\frac{1}{2}$, and password value $v = 1.5- \epsilon$, hash cost $c = 1$. In this case, the attacker will not try to check any password (if only one guess is made the attacker will obtain expected gain $v/2$ with cost 1, resulting in negative expected utility; if the attacker guesses at most twice, gain is $v$ and cost is 1.5, the expected utility is still negative). On the other hand, if breakpoints are time-even, we have $c_1 = 0.4$ (hash cost of evaluating MHF to the first label), $c_2 = 1.6$ (hash cost of evaluating MHF to the second label) with expected cost being $c$; checking the first label of both passwords gives the attacker utility $0.05-\epsilon/2$, as long as $\epsilon < 0.1$ the attacker will crack the password with probability at least $\frac{1}{4}$. 

\paragraph{Cost-Even Breakpoints and Uniform Distribution.}
We have proved that cost-even breakpoints with uniform distribution always yields lower (or equal if no passwords are cracked) attacker's success rate in Theorem \ref{theorem:costeven}. Experiment results (see Figure \ref{fig:costopt}) also verified our therotical conclusion. 
 
Another observation is that increasing $m$ can further reduce the percentage of cracked passwords. In confident region there is up to 10\% of passwords that would have been cracked can be now saved when $m=99$. However, the marginal benefit is diminishing, that is to say, the reduction in $\lambda(\pi^*)$ from $m=i$ to $m=i+1$ is decreasing as $i$ grows larger. From Figure \ref{fig:costopt}, we see that success rate reduction by increasing $m$ from 3 to 7 is roughly the same as that by increasing $m$ from 7 to 99. In full version of this paper, we prove that the marginal cost of checking $\varpi_{i}(1,m)$ given $\oplus_{i^{\prime}=1}^{i-1}\varpi(1, m)$ already been checked is roughly proportional to $\frac{m}{m+1}$, thus the advantage of increasing $m$ is significant when $m$ is small, but there is barely any improvement by increasing $m$ when $m$ is already large. Since $\frac{7}{8} - \frac{3}{4} \approx \frac{99}{100} - \frac{7}{8}$, performance improvement of increasing $m$ from 3 to 7 is roughly the same as that of increasing $m$ from 7 to 99.

% \paragraph{Comparision of Cost-Even and Time-Even Breakpoints.}
% In Figure \ref{fig:compare} of appendix, we present the plots comoparing cost-even and time-even breakpoints under uniform distribution when $m=3$. The performance of two breakpoints settings within the confidenet zone is very close when $v/C_{max} > 10^4$, but time-even breakpoints would benefit the attacker when $v/C_{max}$ is small while cost-even breakpoints would only reduce attacker's success rate.

\paragraph{Optimized Distribution and Cost-Even Breakpoints.}

We use BITEOPT\cite{biteopt2021} to find the optimal distribution $ \vec{q}^*$, the results are shown in Figure \ref{fig:costopt}. It can be seen that the performance of optimal distribution is almost identical to that of uniform distribution in term of attacker's success rate. This can be observed in all of the datasets in our experiments.  Adopting $\vec{q}^*$ has  comparatively small advantage and require the server to know password distribution in advance, which is difficult in practice. Thus we recommend to use cost-even break points with uniform distribution as a \emph{prior independent} password hashing mechanism.

\section{Conclusion}
In this paper, we introduce cost-asymmetric memory hard password authentication, a prior independent  authentication mechanism, to defend against offline attacks. As traditional hash function are replaced by memory hard functions, we propose to use random breakpoints in evaluation of an MHF in order to have the benefit of both cost asymmetry and cost quadratic scaling.  The interaction between the defender and the attacker is modeled by a Stackelberg game, within the game theory framework we formulate the optimal strategies for both defender and attacker. We theoretically proved that cost-asymmetric memory hard password authentication with cost-even breakpoints sampled from uniform distribution will reduce attacker's cracking success rate. In addition we set up experiments to validate the effectiveness of our proposed mechanism for arbitrary parameter settings, experiment results show that the reduction of attacker's success rate is up to 10\%.

% \subsection{Questions}
% \begin{enumerate}
%     \item
% Finding the optimal checking sequence for the attacker given a specific pepper distribution $\vec{q}$ is ensentially the problem of \emph{subset sum with externalities} ($\mathsf{SSE}$): given a set of items $a_1,\ldots, a_n$, select a subset $S$ such that $\mathsf{value}(S) = \sum_{a_i\in S}\mathsf{value}(a_i)$ is maximized where $\mathsf{value}(a_i)$ is function of $S/a_i$, instead of a constant (which is the assumption of classic subset sum). It is clear $\mathsf{SSE} \in \mathsf{EXP}$, can we prove $\mathsf{SSE}$ is NP-complete or deduce other complexity results?
% \item Would our algorithm $\FG{\pi_{LO}}$ work for arbitary $m$?
% \end{enumerate}

\bibliographystyle{splncs04}
{\footnotesize
\bibliography{bib/signal,bib/abbrev0,bib/crypto,bib/bounded-parallel-mhf,bib/jit,bib/extra}
}

\appendix

\subsection{Marginal Cost}
\label{app:marginalcost}
Let \small$\sum_{i^{\prime = 1}}^{i-1}\Pr[pw_i] = \lambda$, $c_M(t_i^2-t_{i-1}^2) = c$\normalsize,  then 
\small
\begin{equation}
\begin{aligned}
&C_{adv}^u(\varpi_{i}(1,m) |\Pi(i-1,m)) = (1-\lambda) c \\
& + \left(1- \lambda - \Pr[pw_i]\frac{1}{m} \right)c +,\ldots,\\
& + \left(1- \lambda - \Pr[pw_i]\frac{m-1}{m} \right)c\\
&= (1-\lambda)mc - \frac{\Pr[pw_i]c(m-1)}{2}.
\end{aligned}
\end{equation}
\normalsize
Because of workload constraint we have,
\small
\begin{equation*}
\alpha = \sum_{i=1}^m \beta_i^2/m = \frac{m+1}{2}.
\end{equation*}
\normalsize
thus, $c = C_{max} / \alpha  = \frac{2C_{max}}{m+1}$, substitute it into equation (16), we have 

\small\begin{equation}
\label{eq:mc1}
\begin{aligned}
&C_{adv}^u(\varpi_{i}(1,m) |\Pi(i-1,m)) \\
&= \frac{2m}{m+1}(1-\lambda) C_{max} - \frac{(m-1)\Pr[pw_i]C_{max}}{m+1}.\\
\end{aligned}
\end{equation}
\normalsize

On the other hand, we have 
\small\begin{equation}
\label{eq:mc2}
C_{adv}^d(pw_i |\oplus_{i^{\prime}=1}^{i-1} pw_{i^{\prime}}) = (1-\lambda)C_{max}.
\end{equation}
\normalsize
Take the difference of equation (\ref{eq:mc1}) and equation (\ref{eq:mc2}),
\small
\begin{equation*}
\begin{aligned}
&C_{adv}^u(\varpi_{i}(1,m) |\Pi(i-1,m)) -  C_{adv}^d(pw_i |\oplus_{i^{\prime}=1}^{i-1} pw_{i^{\prime}}) \\
&= \frac{m - 1}{m+1}C_{max} \left((1-\lambda) -\Pr[pw_i]\right) \geq 0.\\
\end{aligned}
\end{equation*}
\normalsize

\subsection{Optimality Test}
Since $\pi_{LO}$ is locally optimal adding any instruction bundle for a single password $e = \bundle$ into $\pi_{LO}$ will decrease the overall utility, namely, $\mar{e}{\pi}<0$, recall that $\mar{e}{\pi}$ denotes the marginal utility by including $e$ into $\pi$, either through concatenation or insersion.    
However, there is no guarantee $\mar{S}{\pi_{LO}} < 0$ where $S$ is an ordered set of instruction bundles $\{e_1, e_2, \ldots, e_b\}$ since marginal utility is not \emph{additive} with respect to instruction bundles. To see this, from equation \eqref{eq:utility} we can derive 
\begin{equation}\small
\label{eq:marginins}
\mar{S}{\pi_{LO}} = \sum_{e \in S} \mar{e}{\pi_{LO}} +  \sum_{e_2\in S}\overbrace{c(e_2)\sum_{\substack{e_1\in S \\ e_1<e_2}} \Pr[e_1]}^{\text{cost reduction for $e_2$}},
\end{equation}
where $\Pr[e]$ and $c(e)$ are probability summation and round cost summation of instructions in $e$, respectively. Equation \eqref{eq:marginins} shows that the marginal cost by including a ordered set $S$ to $\pi_{LO}$ stems from 2 parts. The first is the summation of individual contribution and the second is cost reduction when checking $e_2\in S$ because previously included instruction bundles $e_1<e_2, e_1\in S$ have already eliminate some uncertainty. Even though every instruction bundle solely contributes negative marginal utility i.e, $\mar{e}{\pi_{LO}} < 0,\, e\in S$, the sign of $\mar{S}{\pi_{LO}}$ is not decisive because of the cost reduction term. If a set $S$ exists such that $\mar{S}{\pi}\geq 0 $, we will refer it to as a \emph{good set}.

By definition of good set and Theorem \ref{theorem:EX}, we have  
\begin{equation}\label{eq:optseq}\small
\pi^* = \pi_{LO} \cup S^*,\text{ s.t. } S^* = \argmax_S \mar{S}{\pi_{LO}}.
\end{equation}
Recall that $\cup$ denotes inclusion (whether through concatenation or insertion) while maintaining natural ordering. \footnote{Technically, $\argmax_S$ returns a set of solutions $S^*$. However, if this set contains multiple elements we can break ties according to the size of $|S^*|$ and followed by an arbitrary lexicographic ordering over solutions with same size.}

Verify if $S$ is a good set of $\pi_{LO}$ is easy but find one is hard. We design a polynomial algorithm to check if the local optimum $\pi_{LO}$ is in fact a global optimum, i.e., $\pi_{LO} = \pi^*$.  Our algorithm utilizes the following observation.
\begin{observation}
if $S = \{e_1, \ldots, e_b\}$ is a good set for $\pi_{LO}$, then its last element $e_b$ must provide non-negative utility (otherwise, it can be safely removed from $S$ without hurting marginal utility), namely
\small$$\mar{e_b}{\pi_{LO} \cup {S\setminus e_b} }\geq 0,$$\normalsize
where $S\setminus e_b$ is the ordered set excluding $e_b$.
\end{observation}

It is not clear which elements are inside $S\setminus e_b$ but we know $S\setminus e_b \subseteq \before{e_b}{\pi_{LO}}$ where $ \before{e_b}{\pi_{LO}}$ is the ordered set of all unchecked instructions that appear before $e_b$ in natural ordering, given $\pi_{LO}$ already being checked, namely, 
\small
$$\before{e_b}{\pi_{LO}}\defeq \{e:e<e_b \text{ and } e\not\in \pi_{LO}\}.$$
\normalsize
We use following Lemma to negate the existence of a good set ending with $e_b$.
\begin{restatable}{lemma}{restest}
	\label{lemma:test}
	For a unchecked instruction bundle $e_b$, define
	\small
	\begin{equation*}
	\mathsf{test}(e_b) \defeq \mar{e_b}{\pi_{LO}}  + \sum_{e\in \before{e_b}{\pi_{LO}}} \Pr[e]c(e_b),
	\end{equation*}
	\normalsize 
	if $\mathsf{test}(e_b) < 0$ then a good set $S$ for $\pi_{LO}$ ending with $e_b$ does not exist.
\end{restatable}

\begin{proof}
	\small
	\begin{equation*}
	\begin{aligned}
	&\mar{e_b}{\pi_{LO} \cup {S\setminus e_b} } =  \mar{e_b}{\pi_{LO}} + \sum_{e\in S\setminus e_b} \Pr[e]c(e_b),\\
	&\leq \mar{e_b}{\pi_{LO}}  + \sum_{e\in \before{e_b}{\pi_{LO}}} \Pr[e]c(e_b)\\
	&=\mathsf{test}(e_b).
	\end{aligned}
	\end{equation*}
	\normalsize
	If $\mathsf{test}(e_b) < 0$, then  $\mar{e_b}{\pi_{LO} \cup {S\setminus e_b} } < 0$, $S$ cannot be a good set.
\end{proof}

We can interpret the term $\sum_{e\in \before{e_b}{\pi_{LO}}} \Pr[e]c(e_b)$ as the maximum possible cost reduction when checking $e_b$, then $\mathsf{test}(e_b)$ is the maximum marginal utility $e_b$ can provide as the last instruction bundle in $S$. If $\mathsf{test}(e_b) < 0$, a set ending with $e_b$ cannot be a good set; if this is the case for all instructions $e_b \in \Pi(\lenmax , m)\setminus \pi_{LO}$ that might be added to $\pi_{LO}$, then a good set ending with any unchecked instruction bundle does not exist. Thus, in Equation \eqref{eq:optseq} we have $S^* = \emptyset$ and $\pi_{LO} = \pi^*$.  We use $\OT{\pi_{LO}}$ to examine if $\pi_{LO} = \pi^*$ --- see Algorithm \ref{algo:OT}.

\begin{algorithm}[h]\small
	\KwIn{$v$, $\vec{q}$, $\pi_{LO}$}
	\KwOut{$\pi^*$}
	\BlankLine
	\ForEach{$e_b\in \Pi(\lenmax , m)\setminus \pi_{LO}$}
	{
		\If{$\mathsf{test}(e_b)\geq 0$}
		{
			\KwRet{FAIL}\;
		}
	}
	\KwRet{PASS}
	\caption{$\OT{\pi_{LO}}$}
	\label{algo:OT}
\end{algorithm}

\begin{restatable}{theorem}{restateoptimalitycheck} 
	\label{thm:optimalitycheck}
If $\OT{\pi_{LO}}$ returns PASS, then $\pi_{LO} = \pi^*$.
\end{restatable}

\begin{proof}
	Since $\OT{\pi_{LO}}$ returns PASS, then we have 
	$$\mar{e_b}{\pi_{LO} \cup {S\setminus e_b} }\leq \mathsf{test}(e_b) < 0,\,  \forall S.$$
	For any $S$ it would only increase overall utility by repeatedly removing the last element in $S$ until it is empty. Therefore, a good set is a empty set. In other words, $\pi_{LO} = \pi^*$.
\end{proof}

By storing of $\lambda(\pi_{LO}, i)$ and $\sum_{i^{\prime}=1}^{i}\Pr[pw_{i^{\prime}}]$ for all $i$, which are intermidiate values in execution of $\EX$, in amortized sense we can evaluate $\mathsf{test}(e_b)$ in time $\mathcal{O}(1)$, then $\OT{\pi_{LO}}$ runs in time $\mathcal{O}(n_pm)$.
We remark that even if $\pi_{LO}$ fails $\OT{\pi_{LO}}$, it does not imply $\pi_{LO} \neq \pi^*$, since $\mathsf{test}(e_b) < 0$ is a sufficient condition of no good set ending with $e_b$, not a necessary one. Fortunately, $\OT{\pi_{LO}}$ returns PASS for most of the time in our experiments, confirming $\pi_{LO} = \pi^*$ for most $v/C_{max}$ ratios; otherwise, we might discard $\OT{\pi_{LO}}$ as well.

\subsection{Finding $\pi^*$ for cost-even breakpoints when $m\leq 3$}

If $\pi_{LO}$ fails $\OT{\pi_{LO}}$, we cannot deduce any conclusions about the optimality of $\pi_{LO}$, but when $\beta_i = \sqrt{i}$ and $m\leq 3$ we can design an efficient brute force algorithm to find $\pi^*$.
\begin{definition}
Given a vector of real numbers  $\vec{q}=(q_1,\ldots, q_m)$ we call $q_j$ a peak of $\vec{q}$ if (1) $j=m$, or (2) $j=1$ and $q_1 > q_2$, or (3) $q_{j-1} \leq q_j$ and $q_j > q_{j+1}$. We use $\peak$ to denote the set of all peak indices in $\vec{q}$. 	 
\end{definition}

\begin{restatable}{lemma}{respeak}
\label{lemma:peak}
Fix an arbitrary breakpoint distribution $\vec{q} = (q_1,\ldots, q_m)$. Suppose that $\beta_j=\sqrt{j}$ for all $j \leq m$ and  $\pi_{LO} = \op{\len{\pi_{LO}}}$, then for all $i \leq \len{\pi_{LO}}$ we have  $\tau_i \in \peak$.
\end{restatable}

\begin{proof}
	Suppose $\tau_i = j \neq m$, since $(pw_i,t_j)\in \pi_{LO}$, then 
	\small
	$$\Delta\left(\pi_{LO}-(pw_i, t_j), \pi_{LO}\right)\geq 0,$$
	\normalsize
	where $-$ is removal operation. Since $j$ is not a peak, then we have $q_{j+1} \geq q_j$ which leads to 
	\small
	$$\margin{(pw_i, t_{j+1})}{\pi_{LO}} > \margin{\pi_{LO}-(pw_i, t_j)}{\pi_{LO}}.$$
	\normalsize
	Therefore $\Delta\left(\pi_{LO}, \pi_{LO} + (pw_i, t_{j+1})\right)>0$. It is still profitable to advance a label for $pw_i$ i.e., check $(pw_i, t_{j+1})$, so $(pw_i,t_{j+1})$ should have been included into $\pi_{LO}$ in local search. Contradiction.
\end{proof}

\begin{restatable}{lemma}{respeakTwo}
\label{lemma:peakTwo}
Fix an arbitrary breakpoint distribution $\vec{q} = (q_1,\ldots, q_m)$. Suppose that $\beta_j=\sqrt{j}$ for all $j \leq m$ and  $\pi^* = \oplus_{i=1}^{\len{\pi^*}} \varpi_i(1, \tau_i^*)$, then for all $i \leq \len{\pi^*}$ we have  $\tau_i^*\in \peak$.
\end{restatable}

\begin{observation}
when $m=2$, $\peak$ is in $\left\{\{2\}, \{1, 2\}\right\}$;
when $m=3$, $\peak$ is in $\left\{\{3\}, \{1, 3\}, \{2, 3\}\right\}$;
\end{observation}
$\lvert \peak\rvert = 1$ correspons to uniform breakpoint distribution for which $\EXC{\emptyset}$ already gives the optimal checking sequence $\pi^*$.

Lemmas \ref{lemma:peak} and \ref{lemma:peakTwo} imply that if a good set $S$ \emph{with respect to insertion} exists for $\pi_{LO}$, every instruction bundle in $S$ must start at a peak position and end with another peak position. Specially, when $\peak = \{peak_1,m\}$, the only tentative insertion operation to promote locally optimal $\pi_{LO}$ to global optimal $\pi^*$ is to check a password to completion, i.e., change the largest label index $\tau_i$ from $peak_1$ to $m$.

The following theorem states that we can efficiently construct an ordered set $S_c(i, \pi)$ which provides larger marginal utility upon insertion than any other ordered set of the same size. 

\begin{restatable}{theorem}{resgoodset}
\label{theorem:goodsetcandidate}
If $\peak = \{peak_1, m\}$ and $\beta_i = \sqrt{i}$ for $i\leq m$,  given $\pi = \op{\lenz}$ with $\tau_i\in \peak$ and $\unchecked{\pi} = \{\varpi_i(peak_1,m):i\leq \len{\pi} \text{ and }\tau_i = peak_1\}$---the set of unchecked instruction bundles spanning two peaks, we define
\small
\begin{equation*} 
	S_c(i, \pi) \defeq
	\begin{cases}
		\emptyset, \text{ if } i=0,\\
		\{e_1, \ldots, e_i\}, \text{ if } i>0.\\
	\end{cases}
\end{equation*}
\normalsize
where $e_i,\, \forall i>0$ is recursively defined as 
\small$$e_i \defeq \argmax_{e \in \unchecked{\pi + S_c(i-1,\pi)}} \ins{e}{\pi + S_c(i-1,\pi)}.$$\normalsize
Then we have
\small
\begin{equation*}
	\begin{aligned}
	\ins{S_{c}(i, \pi)}{\pi} &\geq \ins{S}{\pi}, \\ 
	&\forall i,\, \forall S\subseteq \unchecked{\pi} \text{ s.t. } |S_c(i,\pi)| = |S|.
	\end{aligned}
\end{equation*}
\normalsize
\end{restatable}

We first prove the following Lemma, which is utilized in proof of Theorem \ref{theorem:goodsetcandidate}.

\begin{lemma}
	\label{lemma:marginlambda}
	When $\beta_i=\sqrt{i}$, suppose $e_1 = \varpi_{i1}(j_1,j_2)$, and  $e_2 = \varpi_{i_2}(j_1,j_2)$, if $\ins{e_1}{\pi} \geq \ins{e_1}{\pi}$, then $\Pr(e_1) \geq \Pr(e_2)$.
\end{lemma}
\begin{proof}
	Proof by contradiction. Let $c_M(t_{j_2}^2-t_{j_1}^2) = c$. Suppose $\Pr(e_1) < \Pr(e_2)$, then 
	$e_1 > e_2$, we have 
\small
\begin{equation*}
	\begin{aligned}
	&\ins{e_1}{\pi} - \ins{e_2}{\pi}\\
	&= \Pr(e_1)\left(v + \sum_{e > e_1, e\in \pi}c\right) - \left(1-\sum_{e<e_1, e\in \pi}\Pr(e)\right)c\\ 
	&- \Pr(e_2)\left(v + \sum_{e > e_2, e\in \pi}c\right) + \left(1-\sum_{e<e_2, e\in \pi}\Pr(e)\right)c\\
	&=  \left(\Pr(e_1) - \Pr(e_2)\right)\left(v + \sum_{e>e_1, e\in \pi} c\right)\\ 
	&- \Pr(e_2)\sum_{e_2<e\leq e_1, e\in \pi}c + c\sum_{e_2\leq e<e_1, e\in \pi}\Pr(e)\\
	& =\left(\Pr(e_1) - \Pr(e_2)\right)\left(v + \sum_{e>e_1, e\in \pi} c\right)\\
	& + c\sum_{e_2<e<e_1} \Pr(e) - \Pr(e_2)\\
	& \leq \left(\Pr(e_1) - \Pr(e_2)\right)\left(v + \sum_{e>e_1, e\in \pi} c\right) < 0,
	\end{aligned}
\end{equation*}
\normalsize
which contradicts the precondition $\ins{e_1}{\pi} \geq \ins{e_2}{\pi}$
\end{proof}

Now we prove Theorem \ref{theorem:goodsetcandidate}.
\begin{proof}
	In this proof $S_c(i, \pi)$ is written in $S_c(i)$ for simplicity.
	Let $S = \{e_1^{\prime}, e_2^{\prime}, \ldots, e_i^{\prime}\}$.
	Since $S\subseteq \unchecked{\pi}$ and $S_c(i)\subseteq \unchecked{\pi}$, instruction bundles in $S$ and $S_c(i)$ have the same size. Let $c(e) = c, \forall e\in S_c(i),\, \forall e\in S$. Therefore, by definition of $e_i$ and Lemma \ref{lemma:marginlambda} we have $\Pr(e_j)\geq \Pr(e_j^{\prime}),\, \forall j\leq i$.
	 
	Theorem \ref{theorem:goodsetcandidate} equivalently claims
	\small
	\begin{equation*}
		\label{eq:goodset}
	\ins{S_c(i)}{\pi} \geq \ins{S_c(i-n)\cup \{e_{i-n+1}^{\prime},\ldots,e_i^{\prime}\}}{\pi}, \forall n.
	\end{equation*}
	\normalsize
	We will use mathematical induction to prove the above inequalities.
	Base case $n=1$:
	\small
	\begin{equation*}
		\begin{aligned}
		&\ins{S_c(i)}{\pi}\\
		&= \ins{S_c(i-1)}{\pi} + \ins{e_i}{\pi + S_c(i-1)}\\
		& \geq \ins{S_c(i-1)}{\pi} + \ins{e_i^{\prime}}{\pi + S_c(i-1)}\\
		& = \ins{S_c(i-1) \cup \{e_i^{\prime}\}}{\pi}.
		\end{aligned}
	\end{equation*}
	\normalsize
The inequality holds because of the definition of $e_i$ i.e., $e_i = \arg\max_{e} \ins{e}{\pi + S_c(i-1)}$. 

Inductive hypothesis: equation \eqref{eq:goodset} holds true for $n=k$.

Inductive step: in the following we prove that equation \eqref{eq:goodset} holds true for $n=k+1$.
\small
\begin{equation*}
	\begin{aligned}
		&\ins{S_c(i)}{\pi}\\
		&\geq \ins{S_c(i-k)\cup \{e_{i-k+1}^{\prime},\ldots,e_i^{\prime}\}}{\pi}\\
		& = \ins{S_c(i-k-1)}{\pi} + \ins{e_k}{\pi + S_c(i-k-1)} \\
		& + \ins{\{e_{i-k+1}^{\prime},\ldots,e_i^{\prime}\}}{\pi + S_c(i-k-1) + e_k}\\
		& \geq \ins{S_c(i-k-1)}{\pi} + \ins{e_k^{\prime}}{\pi + S_c(i-k-1)}\\
		& + \ins{\{e_{i-k+1}^{\prime},\ldots,e_i^{\prime}\}}{\pi + S_c(i-k-1) + e_k}\\
		& \geq \ins{S_c(i-k-1)}{\pi} + \ins{e_k^{\prime}}{\pi + S_c(i-k-1)}\\
		& + \ins{\{e_{i-k+1}^{\prime},\ldots,e_i^{\prime}\}}{\pi + S_c(i-k-1) + e_k^{\prime}}\\
		& = \ins{S_c(i-k-1) \cup \{e_{i-k}^{\prime}, \ldots, e_i^{\prime}\}}{\pi}.\\
	\end{aligned}
\end{equation*}
\normalsize
The first inequality is inductive hypothesis; the second inequality holds because of the definition of $e_k$; the third inequality is the result of Lemma \ref{lemma:marginal} given $\Pr[e_k]\geq \Pr[e_k^{\prime}]$.

\end{proof}

We can loop over all candidates $S_c(i, \pi_{LO}), \, \forall i$ and find the \emph{good set w.r.t insertion} $S_c^*(\pi_{LO})$, which provides the largest marginal utility for $\pi_{LO}$, namely,
\small\begin{equation}
	S_c^*(\pi_{LO}) = \argmax_{S_c(i,\pi_{LO})} \ins{S_c(i, \pi_{LO})}{\pi_{LO}}.
\end{equation}\normalsize

A good set $S^*$ might contain instruction bundles that can be concatenated to $\pi_{LO}$, to handle this case we also need to loop over $len$---possible length of $\pi^*$. The efficient brute force algorithm $\FG{\pi}$ is present in Algorithm \ref{algo:FG}.

\begin{algorithm}[h]\small
		\KwIn{$v$, $\vec{q}$, $\pi$}
		\KwOut{$\pi^*$}
		\BlankLine
		
		$U^* = \margin{\emptyset}{\pi}$\;
		\For{$len = \len{\pi}: \lenmax$}
		{
			\If{$len > \len{\pi}$}
			{
				$\pi = \pi\circ \varpi_{len}(1,peak_1)$\;
			}
			\ForEach{$i$}
			{
				construct $S_c(i, \pi)$\;
				compute $\ins{S_c(i,\pi}{\pi}$\;
			}
			$S_c^*(\pi) = \argmax_i \ins{S_c(i, \pi}{\pi}$\;
			\If{$\margin{\emptyset}{\pi + S_c^*(\pi)} \geq U^*$}
			{
				$\pi^* = \pi + S_c^*(\pi)$\;
				$U^* = \margin{\emptyset}{\pi + S_c^*(\pi)}$\;
			}
		}
		\KwRet{$\pi^*$}
		\caption{$\FG{\pi}$}
		\label{algo:FG}
\end{algorithm}

Suppose $|\unchecked{\pi}| = n$, finding $e_1$ takes time $\mathcal{O}(n)$; finding $e_2$ takes time $\mathcal{O}(n-1)$, etc. Thus, the inner loop in Algorithm \ref{algo:FG} takes time $\mathcal{O}(n^2)\subseteq \mathcal{O}(n_p^2)$,  and the total running time is $\mathcal{O}(n_p^2\left(\lenmax -\len{\pi}\right)$. \footnote{In implementation we can exploit opt restriction 3 and the fact that ordering of $e_i$ is mainly unchanged after update $\len{\pi}$ to further reduce running time. We also remark in experiments it is (almost) always the case $\pi^* = \pi_{LO} + S_c^*(\pi_{LO})$ and for most of the time $S_c^*(\pi_{LO}) = \emptyset$.}

\begin{restatable}{theorem}{restatefg}
	\label{theorem:fg}
When $|\peak| = 2$ and $\beta_i = \sqrt{i}$, $\FG{\pi_{LO}}$ returns the optimal checking sequence.
\end{restatable}

\begin{proof}
		Lemma \ref{lemma:peak} and \ref{lemma:peakTwo} restrict a good set to be a subset of $\unchecked{\pi}$.
		 Theorem \ref{theorem:goodsetcandidate} guarantees that $S_c(i, \pi)$ is ``better'' than any other set $S$ of the same size, namely,
		 \small
		 $$
		 \margin{\emptyset}{\pi+S_c(i, \pi)} \geq \margin{\emptyset}{\pi+S}, \forall S \subseteq \unchecked{\pi}.
		 $$
		 \normalsize
		 By definition of $S_c^*(\pi)$, we have
		 \small
		 $
		 \margin{\emptyset}{\pi+S_c^*(\pi)} \geq  \margin{\emptyset}{\pi+S_c(i, \pi)}, \forall i.
		 $
		 \normalsize
		 The outer loop of $\FG{\pi}$ traverses all possible $\len{\pi^*}$ and returns the $\pi+S_c^*(\pi)$ with largest utility. By equation \eqref{eq:optseq} the the returned checking sequence is optimal.
\end{proof}

We could potentially run $\FG{\emptyset}$ to find the optimal checking sequence $\pi^*$. As a shortcut, we run $\FG{\pi_{LO}}$ instead to reduce the running time.
\begin{corollary}
	When $m\leq 3$ and $\beta_i=\sqrt{i}$, There are polynomial algorithms that always find the optimal checking sequence.
\end{corollary}
If $|\peak| = 1$, $\EXC{\emptyset}$ returns the optimal checking sequence; if $|\peak|=2$, $\FG{\pi_{LO}}$ returns the optimal checking sequence.

\subsection{Proof of Other Lemmas}

\subsubsection{Proof of Lemma \ref{lemma:marginal} }
\resmarginal*
\begin{proof}
Suppose $e_i$ is the $i$th instruction bundle in $\pi_3$, from equation \ref{eq:con} we have
\small
\begin{equation*}
	\con{e_1}{\pi_1} = \Pr[pw_i]q_j v - \left(1-\lambda(\pi_1)\right) c_M(t_j^2 - t_{j-1}^2),
\end{equation*}
\normalsize
and
\small
\begin{equation*}
	\con{e_1}{\pi_2} = \Pr[pw_i]q_j v - \left(1-\lambda(\pi_2)\right) c_M(t_j^2 - t_{j-1}^2).
\end{equation*}
\normalsize
Thus, 
\small
\begin{equation*}
	\con{e_1}{\pi_1} \leq \con{e_1}{\pi_2}, 
\end{equation*}
\normalsize
similarly, 
$$\con{e_i}{\pi_1\circ \oplus_{i^{\prime}=1}^{i-1} e_{i^{\prime}}} \leq \con{e_i}{\pi_2\circ \oplus_{i^{\prime}=1}^{i-1} e_{i^{\prime}}}.$$
Therefore, we have
\small\begin{equation*}
\begin{aligned}
\con{\pi_3}{\pi_1} &= \sum_i \con{e_i}{\pi_1\circ \oplus_{i^{\prime}=1}^{i-1} e_{i^{\prime}}}\\
&\leq \sum_i  \con{e_i}{\pi_2\circ \oplus_{i^{\prime}=1}^{i-1} e_{i^{\prime}}} = \con{\pi_3}{\pi_2}
\end{aligned}
\end{equation*}
\normalsize
\end{proof}

% \subsubsection{Proof of Lemma \ref{lemma:test}}
% \restest*
% \begin{proof}
% \small
% \begin{equation*}
% \begin{aligned}
% &\mar{e_b}{\pi_{LO} \cup {S\setminus e_b} } =  \mar{e_b}{\pi_{LO}} + \sum_{e\in S\setminus e_b} \Pr[e]c(e_b),\\
% &\leq \mar{e_b}{\pi_{LO}}  + \sum_{e\in \before{e_b}{\pi_{LO}}} \Pr[e]c(e_b)\\
% &=\mathsf{test}(e_b).
% \end{aligned}
% \end{equation*}
% \normalsize
% If $\mathsf{test}(e_b) < 0$, then  $\mar{e_b}{\pi_{LO} \cup {S\setminus e_b} } < 0$, $S$ cannot be a good set.
% \end{proof}

% \subsubsection{Proof of Lemma \ref{lemma:peak}}
% \respeak*
% \begin{proof}
% Suppose $\tau_i = j \neq m$, since $(pw_i,t_j)\in \pi_{LO}$, then 
% \small
% $$\Delta\left(\pi_{LO}-(pw_i, t_j), \pi_{LO}\right)\geq 0,$$
% \normalsize
% where $-$ is removal operation. Since $j$ is not a peak, then we have $q_{j+1} \geq q_j$ which leads to 
% \small
% $$\margin{(pw_i, t_{j+1})}{\pi_{LO}} > \margin{\pi_{LO}-(pw_i, t_j)}{\pi_{LO}}.$$
% \normalsize
% Therefore $\Delta\left(\pi_{LO}, \pi_{LO} + (pw_i, t_{j+1})\right)>0$. It is still profitable to advance a label for $pw_i$ i.e., check $(pw_i, t_{j+1})$, so $(pw_i,t_{j+1})$ should have been included into $\pi_{LO}$ in local search. Contradiction.
% \end{proof}

\subsection{Proof of Other Theorems}

\subsubsection{Proof of Theorem \ref{theorem:superset}}
\restatesuperset*
\begin{proof}
	Given $\pi^* = \oplus_{i^{\prime}=1}^{i^*} \varpi_{i^{\prime}}(1, \tau_{i^{\prime}})$, suppose there exists $\varpi_i(1,\tau_i) \subseteq \pi^*$, for some $i > \lenmax, \tau_i > 0$,
	we have 
	\small
	\begin{equation*}
	\begin{aligned}
	&\con{\vp{i}} {\op{i-1}}\\
	&\leq \max_{1\leq j \leq m} \{\con{\varpi_i(1,j)} {\op{i-1}}\}\\
	&\leq \max_{1\leq j \leq m} \{\con{\varpi_i(1,j)}{\Pi(i-1,m)}\}\\&< 0,
	\end{aligned}
	\end{equation*}
	\normalsize
	we can safely remove instructions $\varpi_i(1,\tau_i)$ for all $i > \lenmax $ from $\pi^*$ to obtain another checking sequence that yields a better utility. Contradiction. 
	\end{proof}

\subsubsection{Proof of Theorem \ref{theorem:concat}}
\restateconcat*
\begin{proof}

	Suppose $\pi^* = \op{i^*}$ and $\varpi_i(\tau_i+1, j) \in \pi_{LOC}$ is the first instruction bundle that $\pi^*$ and $\pi_{LOC}$ disagree.
	
	Split $\pi^*$ into two parts $\pi_a^*$ and $\pi_b^*$ where $\pi_a^*$ is the sub-sequence from the beginning of $\pi^*$ to instruction $(pw_i, t_{\tau_i})$ inclusive and $\pi_b^*$ is remaining checking sequence. Formally, 
	
	$$
	\pi_a^*\defeq \oplus_{i^{\prime}=1}^{i} \varpi_{i^{\prime}}(1, \tau_{i^{\prime}}),
	$$ 
	and
	$$
	\pi_b^*\defeq \oplus_{i^{\prime}=i+1}^{i^*} \varpi_{i^{\prime}}(1, \tau_{i^{\prime}}).
	$$
	
	Since $\varpi_i(\tau_i+1, j) \in \pi_{l0}$
	$$ 
	\con{\varpi_i(\tau_i+1,j)}{\pi_a^*} \geq 0.
	$$
	
	The attacker can obtain another checking sequence $\sigma = \pi_a^* \circ \varpi_i(\tau_i+1,j) \circ \pi_b^*$. Note that by Lemma \ref{lemma:marginal}, we have
	\small
	\begin{equation*}
		\con{\pi_b^*}{\pi_a^*\circ \varpi_i(\tau_i+1,j)} > \con{\pi_b^*}{\pi_a^*}
	\end{equation*}
	\normalsize
	Adding both sides of above two inequalities, then both sides of the obtained inequality are added by $\margin{\emptyset}{\pi_a^*}$, we have  
	\small
	\begin{equation*}
	\begin{aligned}
	U_{adv}(v,\vec{q},\sigma) &= \margin{\emptyset}{\pi_a^*} + \con{\varpi_i(\tau_i+1,j)}{\pi_a^*}\\
	& + \con{\pi_b^*}{\pi_a^*\circ \varpi_i(\tau_i+1,j)} \\ 
	& > U_{adv}(v,\vec{q},\pi^*),
	\end{aligned}
	\end{equation*}
	\normalsize
	contracting optimality of $\pi^*$.
	\end{proof}

\subsubsection{Proof of Theorem \ref{theorem:copt}}
\restatecopt*
	\begin{proof} 
		Given these parameters it is easy to verify that 
		\small
		\begin{equation*}
		\begin{aligned}
		&\con{\varpi_i(1,1)}{\op{i-1}}\\
		& < \con{\varpi_i(2,2)}{\op{i-1}\circ \varpi_i(1,1)}\\
		&<\cdots\\
		& < \con{\varpi_i(m,m)}{\op{i-1}\circ \varpi_i(1,m-1)}\\
		\end{aligned}
		\end{equation*}
		\normalsize
		Therefore 
		\small
		\begin{equation*}
		\begin{aligned}
		&\max_{0\leq j\leq m}\left\{ \con{\varpi_i(0,m)}{\op{i-1}}\right\}\\
		&=\max\left\{0,\con{\varpi_i(1,m)}{\op{i-1}} \right\}
		\end{aligned}
		\end{equation*}
		\normalsize
		Algorithm $\EXC{\emptyset}$ will set $\tau_i = m$ for $i \leq \len{\pi_{LOC}}$ and  $\tau_i = 0$ for $i > \len{\pi_{LOC}}$.
		In other words,
		\small
		\begin{equation*}
		\begin{cases}
		\con{\varpi_i(1,m)}{\Pi(i-1,m)} \geq 0, \text{ if } i\leq \len{\pi_{LOC}},\\
		\con{\varpi_i(1,m)}{\Pi(i-1,m)} < 0, \text{ if }i>\len{\pi_{LOC}}.
		\end{cases}
		\end{equation*}
		\normalsize
		Those are also the criterion of defining $\lenmax $ under current parameter settings, hence $\len{\pi_{LOC}}=\lenmax $. Moreover, the superset and subset of $\pi^*$ are identical, i.e., $\Pi(\lenmax , m) = \pi_{LOC}$. Since there are no unchecked instruction bundle $\bundle$ for $i\leq \lenmax $, we have $\pi_{LOC} = \pi^* = \Pi(\lenmax ,m)$.
		\end{proof}

\subsubsection{Proof of Theorem \ref{theorem:costeven}}
\restatecosteven*
\begin{proof}
	When using uniform breakpoint distribution, the execution trace of the attacker's cracking process is similar to deterministic hashing, i.e., sequentially eliminating the possibility of $pw_u = pw_i$ (or verifying it with any luck) for $i = 1, 2, \ldots$. It can be verified (see Appendix \ref{app:marginalcost}) that
	\small
	\begin{equation*}
	C_{adv}^u(\varpi_{i}(1,m) |\Pi(i-1,m)) \geq C_{adv}^d(pw_i |\oplus_{i^{\prime}=1}^{i-1} pw_{i^{\prime}}), \;\forall i,
	\end{equation*}
	\normalsize
	where $C_{adv}^u(\varpi_{i}(1,m) |\Pi(i-1,m))$ is the marginal cost of checking $\varpi_{i}(1,m)$, given $\Pi(i-1,m)$ has already been checked. $C_{adv}^d(pw_i |\oplus_{i^{\prime}=1}^{i-1} pw_{i^{\prime}})$ is the marginal cost of checking $pw_i$ under deterministic cost hashing given passwords $pw_1,\ldots,pw_{i-1}$ have been checked.
	In order to achieve the same success rate, uniform cost-even breakpoints would incur more cost than deterministic cost hashing. See it in another way, when $v/C_{max}$ is fixed for both cases uniform cost-even breakpoints results in a lower adversary success rate.
	\end{proof}

\subsubsection{Proof of Theorem \ref{theorem:EX}}
\restateex*
\begin{proof}
	In proof of Theorem \ref{theorem:concat} we already know $\circ$ operation preserves the following invariant 
	$$ \pi \circ e \subseteq \pi^*, \text{ if } \pi \subseteq \pi^* \text{ and }\con{e}{\pi} \geq 0.$$
	Lemma \ref{lemma:plus} states
	$$ \pi + e \subseteq \pi^*, \text{ if } \pi \subseteq \pi^* \text{ and }\ins{e}{\pi} \geq 0.$$
	$\pi_{LO}$ is obtained by iteratively applying $\circ$ and $+$ operation, hence  is a subset of $\pi^*$
	\end{proof}

\subsection{Derivative-Free Optimization}
There are many derivative-free optimization solvers available in the literature, generally they fall into two catagorites, deterministic algorithm (such as Nelder-Mead) and evolutionary algorithm (such as BITEOPT \cite{biteopt2021} and CMA-EA \cite{Hansen06thecma}). $\mathsf{OptPepperDis}()$ takes password value $v$ as input and outputs optimal pepper distribution $\vec{q}^*$ and attacker's success rate $P_{adv}^*$ when playing with best response given defender's strategy $\vec{q}^*$. During one iteration of  $\mathsf{OptPepperDis}()$, some candidate pepper distributions $\{\vec{q}_{c_i}\}$ are proposed,  together they are referred as \emph{population}. Then the algorithm $\mathsf{BestRes}(v, \vec{q}_{c_i})$  is called as a subroutine for each member of population, and the returned $P_{adv}$ is recorded as ``fitness''. At the end of each iteration, the population is updated according to fitness of its' members, the update could be either through deterministic transformation (Nelder-Mead) or randomized evolution (BITEOPT, CMA-EA). When the iteration number reaches a pre-defined value $ite$, the best fit member $\vec{q}^*$ and its fitness $P_{adv}^*$ are returned.

\subsection{Results of other datasets}
% !TEX root = main.tex
\tikzDefaultsTime
\begin{figure*}[ht]\centering
\subfloat[Bfield]{
%%% General Lower Bound with indegree = 2,3,4,5 and n=2^24, 2^28,..., 2^60
\begin{tikzpicture}[scale=0.45]
\begin{semilogxaxis}[]
%bfield

\addplot file {./newfig/sr/benchmark/bfield_benchmark.dat};

\addplot file {./newfig/sr/timeeven_m3/bfield_timeeven_m3.dat};

\addplot+[only marks] coordinates
{
(300000,0.460015)
(400000,0.505917)
(500000,0.505917)
};

\path[name path=B] (axis cs:(9 * 1e7,0) -- (axis cs:(9 * 1e7,1);
\path[name path=A] (axis cs:(300000,0) -- (axis cs:(300000,1);
\tikzfillbetween[of=A and B, on layer=main]{red, opacity=0.2};

\path[name path=A1] (axis cs:(80000,0) -- (axis cs:(80000,1);
\tikzfillbetween[of=A1 and B, on layer=main]{yellow, opacity=0.2};

\node at (axis cs:5*1e7,0.5) {\rotatebox{90}{\tiny uncertain region}};

\end{semilogxaxis}
\end{tikzpicture}
\label{fig:bfield}
}
\hfill
\subfloat[Brazzers]{
\begin{tikzpicture}[scale=0.45]
\begin{semilogxaxis}
%deterministic
\addplot file {./newfig/sr/benchmark/brazzers_benchmark.dat};

\addplot file {./newfig/sr/timeeven_m3/brazzers_timeeven_m3.dat};

\addplot+[only marks] coordinates
{
(400000,0.534352)
(500000,0.586443)
(600000,0.604296)
(700000,0.604296)
};

\path[name path=B] (axis cs:(9 * 1e7,0) -- (axis cs:(9 * 1e7,1);
\path[name path=A1] (axis cs:(3 *1e5,0) -- (axis cs:(3 * 1e5,1);
\tikzfillbetween[of=A1 and B, on layer=main]{red, opacity=0.2};

\path[name path=A] (axis cs:(2 *1e5,0) -- (axis cs:(2 * 1e5,1);
\tikzfillbetween[of=A and B, on layer=main]{yellow, opacity=0.2};

\node at (axis cs:5*1e7,0.5) {\rotatebox{90}{\tiny uncertain region}};

\end{semilogxaxis}
\end{tikzpicture}
\label{fig:brazzers}
}
\hfill
\subfloat[Clixsense]{
\begin{tikzpicture}[scale=0.45]
\begin{semilogxaxis}
%deterministic
\addplot file {./newfig/sr/benchmark/clixsense_benchmark.dat};

\addplot file {./newfig/sr/timeeven_m3/clixsense_timeeven_m3.dat};

\addplot+[only marks] coordinates
{
(2e+06,0.531809)
};

\path[name path=A] (axis cs:(4 * 1e5,0) -- (axis cs:(4 * 1e5,1);
\path[name path=B] (axis cs:(9 * 1e7,0) -- (axis cs:(9 * 1e7,1);
\tikzfillbetween[of=A and B, on layer=main]{yellow, opacity=0.2};

\path[name path=A1] (axis cs:(8 * 1e5,0) -- (axis cs:(8 * 1e5,1);
\tikzfillbetween[of=A1 and B, on layer=main]{red, opacity=0.2};
\node at (axis cs:5*1e7,0.5) {\rotatebox{90}{\tiny uncertain region}};

\end{semilogxaxis}
\end{tikzpicture}
\label{fig:clixsense}
}

\subfloat[CSDN]{
\begin{tikzpicture}[scale=0.45]
\begin{semilogxaxis}[]
% benchmark
\addplot file {./newfig/sr/benchmark/csdn_benchmark.dat};

\addplot file {./newfig/sr/timeeven_m3/csdn_timeeven_m3.dat};

\addplot+[only marks] coordinates
{
(3e+06,0.555844)
(4e+06,0.597874)
};

\path[name path=A] (axis cs:(2* 1e6,0) -- (axis cs:(2* 1e6,1);
\path[name path=B] (axis cs:(9 * 1e7,0) -- (axis cs:(9 * 1e7,1);
\tikzfillbetween[of=A and B, on layer=main]{red, opacity=0.2};

\path[name path=A1] (axis cs:(1* 1e6,0) -- (axis cs:(1* 1e6,1);
\tikzfillbetween[of=A1 and B, on layer=main]{yellow, opacity=0.2};
\node at (axis cs:5*1e7,0.5) {\rotatebox{90}{\tiny uncertain region}};

\end{semilogxaxis}
\end{tikzpicture}
\label{fig:csdn}
}
\hfill
\subfloat[Rockyou]{
\begin{tikzpicture}[scale=0.45]
\begin{semilogxaxis}
% benchmark
\addplot file {./newfig/sr/benchmark/rockyou_benchmark.dat};

\addplot file {./newfig/sr/timeeven_m3/rockyou_timeeven_m3.dat};

\addplot+[only marks] coordinates
{
(900000,0.327367)
(1e+07,0.690898)
};

\path[name path=A] (axis cs:(7 * 1e6,0) -- (axis cs:(7 * 1e6,1);
\path[name path=B] (axis cs:(9 * 1e7,0) -- (axis cs:(9 * 1e7,1);
\tikzfillbetween[of=A and B, on layer=main]{red, opacity=0.2};

\path[name path=A1] (axis cs:(3 * 1e6,0) -- (axis cs:(3 * 1e6,1);
\tikzfillbetween[of=A1 and B, on layer=main]{yellow, opacity=0.2};

\node at (axis cs:5*1e7,0.5) {\rotatebox{90}{\tiny uncertain region}};

\end{semilogxaxis}
\end{tikzpicture}
\label{fig:rockyou}
}
\hfill
\subfloat[Webhost]{
\begin{tikzpicture}[scale=0.45]
\begin{semilogxaxis}
% benchmark
\addplot file {./newfig/sr/benchmark/webhost_benchmark.dat};

\addplot file {./newfig/sr/timeeven_m3/webhost_timeeven_m3.dat};

\addplot+[only marks] coordinates
{
(8e+06,0.505761)
(9e+06,0.547106)
(1e+07,0.565415)
};

\path[name path=A] (axis cs: 5* 1e6,0) -- (axis cs: 5* 1e6,1);
\path[name path=B] (axis cs: 9 * 1e7,0) -- (axis cs:9 * 1e7,1);
\tikzfillbetween[of=A and B, on layer=main]{red, opacity=0.2};

\path[name path=A1] (axis cs: 2* 1e6,0) -- (axis cs: 2* 1e6,1);
\tikzfillbetween[of=A1 and B, on layer=main]{yellow, opacity=0.2};

\node at (axis cs:5*1e7,0.5) {\rotatebox{90}{\tiny uncertain region}};

\end{semilogxaxis}
\end{tikzpicture}
\label{fig:webhost}
}

\caption{Time-Even Breakpoints, Uniform Pepper Distribution} 
\label{RockYou}
%\vspace{-0.4cm}
\end{figure*}
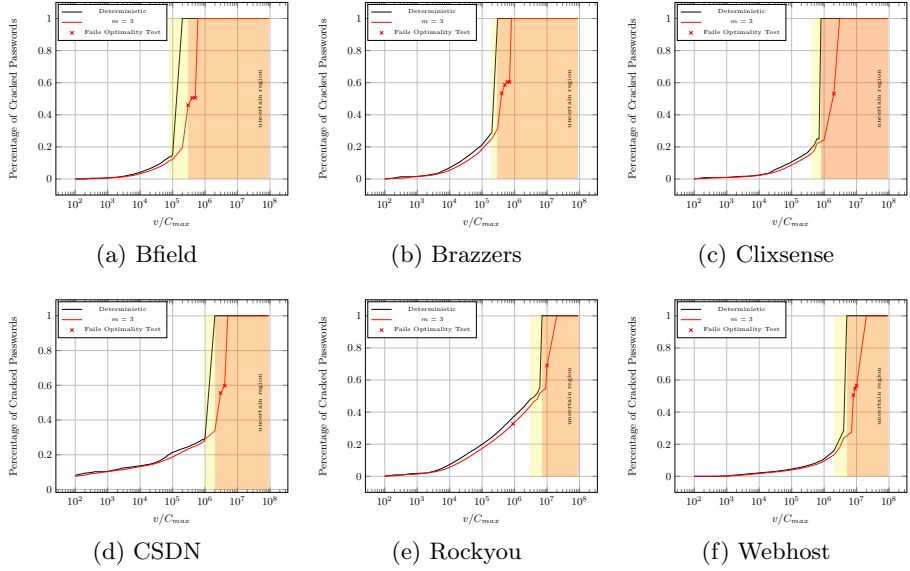
% !TEX root = main.tex
\tikzDefaultsCost
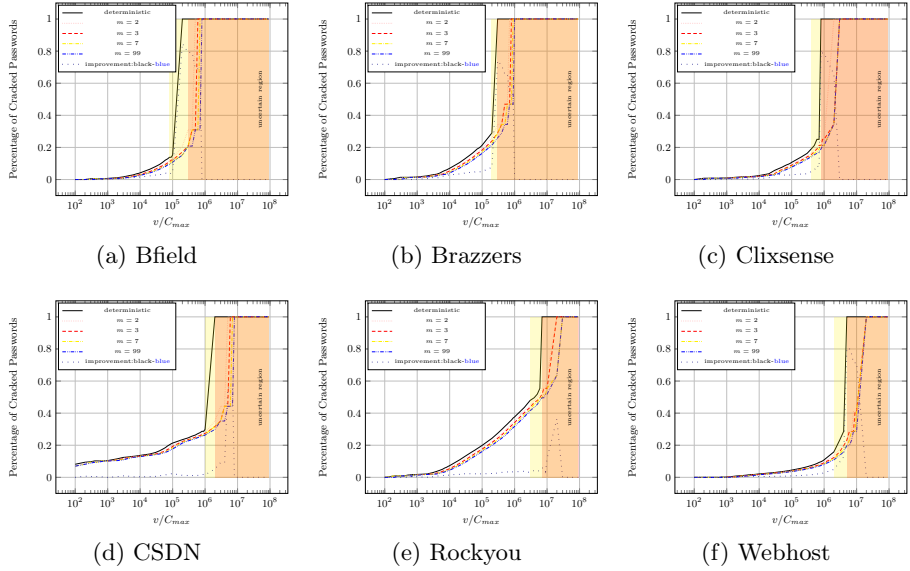
\begin{figure*}[ht]
\centering
\subfloat[Bfield]{
\begin{tikzpicture}[scale=0.45]
\begin{semilogxaxis}[]

\addplot+[stack plots=y] file {./newfig/sr/benchmark/bfield_benchmark.dat};

\addplot file {./newfig/sr/costeven_m2/bfield_costeven_m2.dat};

\addplot file {./newfig/sr/costeven_m3/bfield_costeven_m3.dat};

\addplot file {./newfig/sr/costeven_m7/bfield_costeven_m7.dat};

\addplot file {./newfig/sr/costeven_m99/bfield_costeven_m99.dat};

\addplot+[stack plots=y, stack dir=minus] file {./newfig/sr/costeven_m99/bfield_costeven_m99.dat};

\path[name path=B] (axis cs:(9 * 1e7,0) -- (axis cs:(9 * 1e7,1);
\path[name path=A] (axis cs:(300000,0) -- (axis cs:(300000,1);
\tikzfillbetween[of=A and B, on layer=main]{red, opacity=0.2};

\path[name path=A1] (axis cs:(80000,0) -- (axis cs:(80000,1);
\tikzfillbetween[of=A1 and B, on layer=main]{yellow, opacity=0.2};

\node at (axis cs:5*1e7,0.5) {\rotatebox{90}{\tiny uncertain region}};

\end{semilogxaxis}
\end{tikzpicture}
\label{fig:bfield}
}
\hfill
\subfloat[Brazzers]{
\begin{tikzpicture}[scale=0.45]
\begin{semilogxaxis}
%deterministic
\addplot+[stack plots=y] file {./newfig/sr/benchmark/brazzers_benchmark.dat};

\addplot file {./newfig/sr/costeven_m2/brazzers_costeven_m2.dat};

\addplot file {./newfig/sr/costeven_m3/brazzers_costeven_m3.dat};

\addplot file {./newfig/sr/costeven_m7/brazzers_costeven_m7.dat};

\addplot file {./newfig/sr/costeven_m99/brazzers_costeven_m99.dat};

\addplot+[stack plots=y, stack dir=minus] file {./newfig/sr/costeven_m99/brazzers_costeven_m99.dat};

\path[name path=B] (axis cs:(9 * 1e7,0) -- (axis cs:(9 * 1e7,1);
\path[name path=A1] (axis cs:(3 *1e5,0) -- (axis cs:(3 * 1e5,1);
\tikzfillbetween[of=A1 and B, on layer=main]{red, opacity=0.2};

\path[name path=A] (axis cs:(2 *1e5,0) -- (axis cs:(2 * 1e5,1);
\tikzfillbetween[of=A and B, on layer=main]{yellow, opacity=0.2};

\node at (axis cs:5*1e7,0.5) {\rotatebox{90}{\tiny uncertain region}};

\end{semilogxaxis}
\end{tikzpicture}
\label{fig:brazzers}
}
\hfill
\subfloat[Clixsense]{
\begin{tikzpicture}[scale=0.45]
\begin{semilogxaxis}
%deterministic
\addplot+[stack plots=y] file {./newfig/sr/benchmark/clixsense_benchmark.dat};
%costeven_m2
\addplot file {./newfig/sr/costeven_m2/clixsense_costeven_m2.dat};

\addplot file {./newfig/sr/costeven_m3/clixsense_costeven_m3.dat};

\addplot file {./newfig/sr/costeven_m7/clixsense_costeven_m7.dat};

\addplot file {./newfig/sr/costeven_m99/clixsense_costeven_m99.dat};

\addplot+[stack plots=y, stack dir=minus] file {./newfig/sr/costeven_m99/clixsense_costeven_m99.dat};

\path[name path=A] (axis cs:(4 * 1e5,0) -- (axis cs:(4 * 1e5,1);
\path[name path=B] (axis cs:(9 * 1e7,0) -- (axis cs:(9 * 1e7,1);
\tikzfillbetween[of=A and B, on layer=main]{yellow, opacity=0.2};

\path[name path=A1] (axis cs:(8 * 1e5,0) -- (axis cs:(8 * 1e5,1);
\tikzfillbetween[of=A1 and B, on layer=main]{red, opacity=0.2};
\node at (axis cs:5*1e7,0.5) {\rotatebox{90}{\tiny uncertain region}};

\end{semilogxaxis}
\end{tikzpicture}
\label{fig:clixsense}
}

\subfloat[CSDN]{
\begin{tikzpicture}[scale=0.45]
\begin{semilogxaxis}[]
% benchmark
\addplot file {./newfig/sr/benchmark/csdn_benchmark.dat};

%costeven_m2
\addplot+[stack plots=y] file {./newfig/sr/costeven_m2/csdn_costeven_m2.dat};

\addplot file {./newfig/sr/costeven_m3/csdn_costeven_m3.dat};

\addplot file {./newfig/sr/costeven_m7/csdn_costeven_m7.dat};

\addplot file {./newfig/sr/costeven_m99/csdn_costeven_m99.dat};

\addplot+[stack plots=y, stack dir=minus] file {./newfig/sr/costeven_m99/csdn_costeven_m99.dat};

\path[name path=A] (axis cs:(2* 1e6,0) -- (axis cs:(2* 1e6,1);
\path[name path=B] (axis cs:(9 * 1e7,0) -- (axis cs:(9 * 1e7,1);
\tikzfillbetween[of=A and B, on layer=main]{red, opacity=0.2};

\path[name path=A1] (axis cs:(1* 1e6,0) -- (axis cs:(1* 1e6,1);
\tikzfillbetween[of=A1 and B, on layer=main]{yellow, opacity=0.2};
\node at (axis cs:5*1e7,0.5) {\rotatebox{90}{\tiny uncertain region}};

\end{semilogxaxis}
\end{tikzpicture}
\label{fig:csdn}
}
\hfill
\subfloat[Rockyou]{
\begin{tikzpicture}[scale=0.45]
\begin{semilogxaxis}
% benchmark
\addplot file {./newfig/sr/benchmark/rockyou_benchmark.dat};

%costeven_m2
\addplot+[stack plots=y] file {./newfig/sr/costeven_m2/rockyou_costeven_m2.dat};

\addplot file {./newfig/sr/costeven_m3/rockyou_costeven_m3.dat};

\addplot file {./newfig/sr/costeven_m7/rockyou_costeven_m7.dat};

\addplot file {./newfig/sr/costeven_m99/rockyou_costeven_m99.dat};

\addplot+[stack plots=y, stack dir=minus] file {./newfig/sr/costeven_m99/rockyou_costeven_m99.dat};

\path[name path=A] (axis cs:(7 * 1e6,0) -- (axis cs:(7 * 1e6,1);
\path[name path=B] (axis cs:(9 * 1e7,0) -- (axis cs:(9 * 1e7,1);
\tikzfillbetween[of=A and B, on layer=main]{red, opacity=0.2};

\path[name path=A1] (axis cs:(3 * 1e6,0) -- (axis cs:(3 * 1e6,1);
\tikzfillbetween[of=A1 and B, on layer=main]{yellow, opacity=0.2};

\node at (axis cs:5*1e7,0.5) {\rotatebox{90}{\tiny uncertain region}};

\end{semilogxaxis}
\end{tikzpicture}
\label{fig:rockyou}
}
\hfill
\subfloat[Webhost]{
\begin{tikzpicture}[scale=0.45]
\begin{semilogxaxis}
% benchmark
\addplot+[stack plots=y] file {./newfig/sr/benchmark/webhost_benchmark.dat};

%costeven_m2
\addplot file {./newfig/sr/costeven_m2/webhost_costeven_m2.dat};

\addplot file {./newfig/sr/costeven_m3/webhost_costeven_m3.dat};

\addplot file {./newfig/sr/costeven_m7/webhost_costeven_m7.dat};

\addplot file{./newfig/sr/costeven_m99/webhost_costeven_m99.dat};

\addplot+[stack plots=y, stack dir=minus] file {./newfig/sr/costeven_m99/webhost_costeven_m99.dat};

\path[name path=A] (axis cs: 5* 1e6,0) -- (axis cs: 5* 1e6,1);
\path[name path=B] (axis cs: 9 * 1e7,0) -- (axis cs:9 * 1e7,1);
\tikzfillbetween[of=A and B, on layer=main]{red, opacity=0.2};

\path[name path=A1] (axis cs: 2* 1e6,0) -- (axis cs: 2* 1e6,1);
\tikzfillbetween[of=A1 and B, on layer=main]{yellow, opacity=0.2};

\node at (axis cs:5*1e7,0.5) {\rotatebox{90}{\tiny uncertain region}};

\end{semilogxaxis}
\end{tikzpicture}
\label{fig:webhost}
}
\caption{Cost-Even Breakpoints, Uniform Pepper Distribution} 
\label{RockYou}
\vspace{-0.4cm}
\end{figure*}
% !TEX root = main.tex
\tikzDefaultsBite
\begin{figure*}[ht]\centering
\subfloat[Bfield]{
\begin{tikzpicture}[scale=0.45]
\begin{semilogxaxis}[]

\addplot file {./newfig/sr/benchmark/bfield_benchmark.dat};

\addplot file {./newfig/sr/costeven_m3/bfield_costeven_m3.dat};

\addplot file {./newfig/sr/biteopt_m3/bfield_biteopt_m3.dat};

% \addplot file {./newfig/sr/costeven_m7/bfield_costeven_m7.dat};

\path[name path=B] (axis cs:(9 * 1e7,0) -- (axis cs:(9 * 1e7,1);
\path[name path=A] (axis cs:(300000,0) -- (axis cs:(300000,1);
\tikzfillbetween[of=A and B, on layer=main]{red, opacity=0.2};

\path[name path=A1] (axis cs:(80000,0) -- (axis cs:(80000,1);
\tikzfillbetween[of=A1 and B, on layer=main]{yellow, opacity=0.2};

\node at (axis cs:5*1e7,0.5) {\rotatebox{90}{\tiny uncertain region}};

\end{semilogxaxis}
\end{tikzpicture}
\label{fig:bfield}
}
\hfill
\subfloat[Brazzers]{
\begin{tikzpicture}[scale=0.45]
\begin{semilogxaxis}
%deterministic
\addplot file {./newfig/sr/benchmark/brazzers_benchmark.dat};

\addplot file {./newfig/sr/costeven_m3/brazzers_costeven_m3.dat};

\addplot file {./newfig/sr/biteopt_m3/brazzers_biteopt_m3.dat};
%costeven_m7
% \addplot file {./newfig/sr/costeven_m7/brazzers_costeven_m7.dat};

\path[name path=B] (axis cs:(9 * 1e7,0) -- (axis cs:(9 * 1e7,1);
\path[name path=A1] (axis cs:(3 *1e5,0) -- (axis cs:(3 * 1e5,1);
\tikzfillbetween[of=A1 and B, on layer=main]{red, opacity=0.2};

\path[name path=A] (axis cs:(2 *1e5,0) -- (axis cs:(2 * 1e5,1);
\tikzfillbetween[of=A and B, on layer=main]{yellow, opacity=0.2};

\node at (axis cs:5*1e7,0.5) {\rotatebox{90}{\tiny uncertain region}};

\end{semilogxaxis}
\end{tikzpicture}
\label{fig:brazzers}
}
\hfill
\subfloat[Clixsense]{
\begin{tikzpicture}[scale=0.45]
\begin{semilogxaxis}
%deterministic
\addplot file {./newfig/sr/benchmark/clixsense_benchmark.dat};

\addplot file {./newfig/sr/costeven_m7/clixsense_costeven_m7.dat};

\addplot file {./newfig/sr/biteopt_m3/clixsense_biteopt_m3.dat};

% \addplot file {./newfig/sr/costeven_m7/clixsense_costeven_m7.dat};

\path[name path=A] (axis cs:(4 * 1e5,0) -- (axis cs:(4 * 1e5,1);
\path[name path=B] (axis cs:(9 * 1e7,0) -- (axis cs:(9 * 1e7,1);
\tikzfillbetween[of=A and B, on layer=main]{yellow, opacity=0.2};

\path[name path=A1] (axis cs:(8 * 1e5,0) -- (axis cs:(8 * 1e5,1);
\tikzfillbetween[of=A1 and B, on layer=main]{red, opacity=0.2};
\node at (axis cs:5*1e7,0.5) {\rotatebox{90}{\tiny uncertain region}};

\end{semilogxaxis}
\end{tikzpicture}
\label{fig:clixsense}
}

\subfloat[CSDN]{
\begin{tikzpicture}[scale=0.45]
\begin{semilogxaxis}[]
% benchmark
\addplot file {./newfig/sr/benchmark/csdn_benchmark.dat};

\addplot file {./newfig/sr/costeven_m3/csdn_costeven_m3.dat};

\addplot file {./newfig/sr/biteopt_m3/csdn_biteopt_m3.dat};

% \addplot file {./newfig/sr/costeven_m7/csdn_costeven_m7.dat};

\path[name path=A] (axis cs:(2* 1e6,0) -- (axis cs:(2* 1e6,1);
\path[name path=B] (axis cs:(9 * 1e7,0) -- (axis cs:(9 * 1e7,1);
\tikzfillbetween[of=A and B, on layer=main]{red, opacity=0.2};

\path[name path=A1] (axis cs:(1* 1e6,0) -- (axis cs:(1* 1e6,1);
\tikzfillbetween[of=A1 and B, on layer=main]{yellow, opacity=0.2};
\node at (axis cs:5*1e7,0.5) {\rotatebox{90}{\tiny uncertain region}};

\end{semilogxaxis}
\end{tikzpicture}
\label{fig:csdn}
}
\hfill
\subfloat[Rockyou]{
\begin{tikzpicture}[scale=0.45]
\begin{semilogxaxis}
% benchmark
\addplot file {./newfig/sr/benchmark/rockyou_benchmark.dat};

\addplot file {./newfig/sr/costeven_m3/rockyou_costeven_m3.dat};

\addplot file {./newfig/sr/biteopt_m3/rockyou_biteopt_m3.dat};
%costeven_m7
% \addplot file {./newfig/sr/costeven_m7/rockyou_costeven_m7.dat};

\path[name path=A] (axis cs:(7 * 1e6,0) -- (axis cs:(7 * 1e6,1);
\path[name path=B] (axis cs:(9 * 1e7,0) -- (axis cs:(9 * 1e7,1);
\tikzfillbetween[of=A and B, on layer=main]{red, opacity=0.2};

\path[name path=A1] (axis cs:(3 * 1e6,0) -- (axis cs:(3 * 1e6,1);
\tikzfillbetween[of=A1 and B, on layer=main]{yellow, opacity=0.2};

\node at (axis cs:5*1e7,0.5) {\rotatebox{90}{\tiny uncertain region}};

\end{semilogxaxis}
\end{tikzpicture}
\label{fig:rockyou}
}
\hfill
\subfloat[Webhost]{
\begin{tikzpicture}[scale=0.45]
\begin{semilogxaxis}
% benchmark
\addplot file {./newfig/sr/benchmark/webhost_benchmark.dat};

\addplot file {./newfig/sr/costeven_m3/webhost_costeven_m3.dat};

\addplot file {./newfig/sr/biteopt_m3/webhost_biteopt_m3.dat};
%costeven_m7
% \addplot file {./newfig/sr/costeven_m7/webhost_costeven_m7.dat};

\path[name path=A] (axis cs: 5* 1e6,0) -- (axis cs: 5* 1e6,1);
\path[name path=B] (axis cs: 9 * 1e7,0) -- (axis cs:9 * 1e7,1);
\tikzfillbetween[of=A and B, on layer=main]{red, opacity=0.2};

\path[name path=A1] (axis cs: 2* 1e6,0) -- (axis cs: 2* 1e6,1);
\tikzfillbetween[of=A1 and B, on layer=main]{yellow, opacity=0.2};

\node at (axis cs:5*1e7,0.5) {\rotatebox{90}{\tiny uncertain region}};

\end{semilogxaxis}
\end{tikzpicture}
\label{fig:webhost}
}

\caption{Cost-Even Breakpoints, Optimized Pepper Distribution} 
\label{RockYou}
%\vspace{-0.4cm}
\end{figure*}
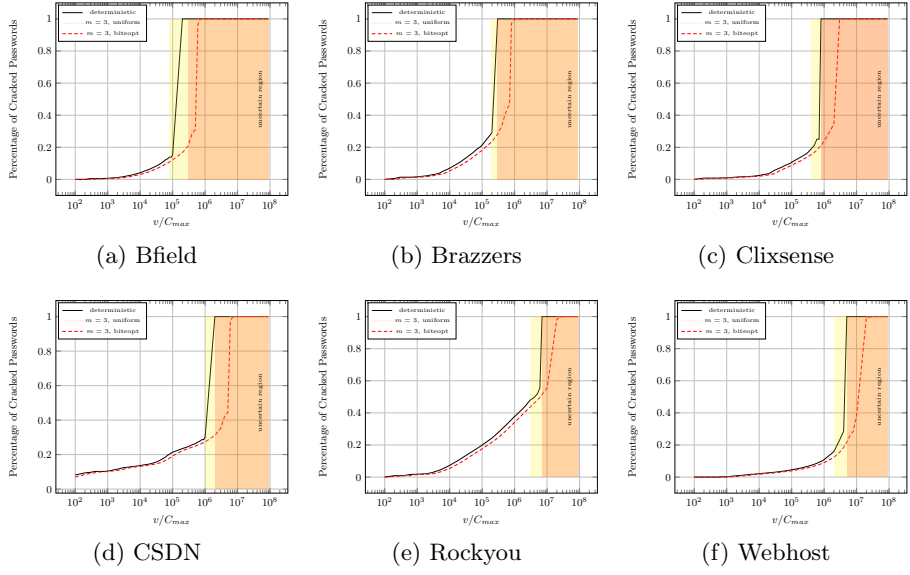
% !TEX root = main.tex
\tikzDefaultsCom
\begin{figure*}[ht]
\centering

\subfloat[Linkedin]{
\begin{tikzpicture}[scale=0.7]
\begin{semilogxaxis}
% benchmark
\addplot+[stack plots=y] file {./newfig/sr/benchmark/linkedin_benchmark.dat};
%costeven_m2
\addplot file {./newfig/sr/timeeven_m3/linkedin_timeeven_m3.dat};

\addplot file {./newfig/sr/costeven_m3/linkedin_costeven_m3.dat};

\path[name path=A] (axis cs:(3* 1e7,0) -- (axis cs:(3* 1e7,1);
\path[name path=B] (axis cs:(9 * 1e7,0) -- (axis cs:(9 * 1e7,1);
\tikzfillbetween[of=A and B, on layer=main]{red, opacity=0.2};

\path[name path=A1] (axis cs:(2* 1e7,0) -- (axis cs:(2* 1e7,1);
\tikzfillbetween[of=A1 and B, on layer=main]{yellow, opacity=0.2};

\node at (axis cs:5*1e7,0.5) {\rotatebox{90}{\tiny uncertain region}};

\end{semilogxaxis}
\end{tikzpicture}
\label{fig:linkedincost}
}
\hfill
\subfloat[Neopets]{
\begin{tikzpicture}[scale=0.7]
\begin{semilogxaxis}
%deterministic
\addplot+[stack plots=y] file {./newfig/sr/benchmark/neopets_benchmark.dat};
%costeven_m2

\addplot file {./newfig/sr/timeeven_m3/neopets_timeeven_m3.dat};

\addplot file {./newfig/sr/costeven_m3/neopets_costeven_m3.dat};

\path[name path=A] (axis cs:(2* 1e7,0) -- (axis cs:(2* 1e7,1);
\path[name path=B] (axis cs:(9 * 1e7,0) -- (axis cs:(9 * 1e7,1);
\tikzfillbetween[of=A and B, on layer=main]{red, opacity=0.2};

\path[name path=A1] (axis cs:(5* 1e6,0) -- (axis cs:(5* 1e6,1);
\tikzfillbetween[of=A1 and B, on layer=main]{yellow, opacity=0.2};

\node at (axis cs:5*1e7,0.5) {\rotatebox{90}{\tiny uncertain region}};

\end{semilogxaxis}
\end{tikzpicture}
\label{fig:neopetscost}
}
% \hfill
% \subfloat[Yahoo]{
% \begin{tikzpicture}[scale=0.7]
% \begin{semilogxaxis}
% % benchmark
% \addplot file {./newfig/sr/benchmark/yahoo_benchmark.dat};

% \addplot file {./newfig/sr/timeeven_m3/yahoo_timeeven_m3.dat};

% \addplot file {./newfig/sr/costeven_m3/yahoo_costeven_m3.dat};

% \path[name path=A] (axis cs:(2* 1e7,0) -- (axis cs:(2* 1e7,1);
% \path[name path=B] (axis cs:(9 * 1e7,0) -- (axis cs:(9 * 1e7,1);
% \tikzfillbetween[of=A and B, on layer=main]{red, opacity=0.2};

% \path[name path=A1] (axis cs:(7* 1e6,0) -- (axis cs:(7* 1e6,1);
% \tikzfillbetween[of=A1 and B, on layer=main]{yellow, opacity=0.2};

% \node at (axis cs:5*1e7,0.5) {\rotatebox{90}{\tiny uncertain region}};

% \end{semilogxaxis}
% \end{tikzpicture}
% \label{fig:yahoocost}
% }
% \caption{Comparision of Time-Even Breakpoints and Cost-Even Breakpoints} 
% \label{fig:compare}
%\vspace{-0.4cm}
\end{figure*}

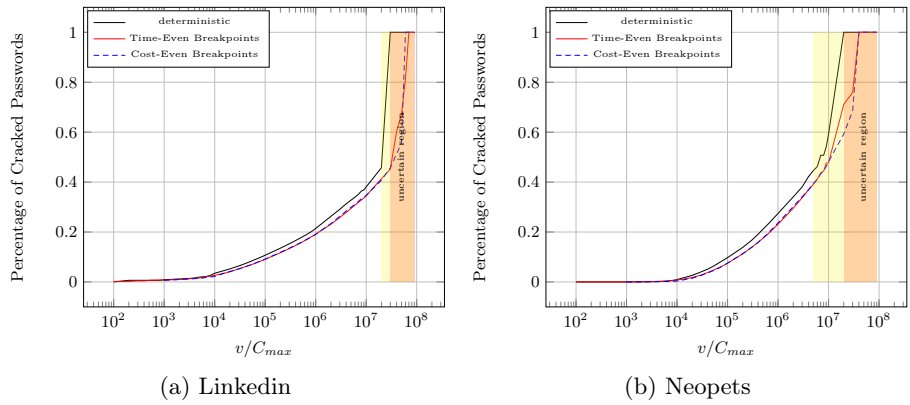
\end{document}